\documentclass[twocolumn,superscriptaddress,prl]{revtex4-2}

\usepackage{graphicx}
\usepackage{amsmath}
\usepackage{amsthm}
\usepackage{amssymb}
\usepackage{latexsym}
\usepackage{array}
\usepackage{hyperref}
\usepackage{amsfonts}
\usepackage{dsfont}
\usepackage{mathrsfs}
\usepackage{verbatim}
\usepackage{bbold}
\usepackage[normalem]{ulem}
\usepackage{upgreek}
\usepackage{makecell}
\usepackage{adjustbox}
\usepackage{algorithm}
\usepackage{algpseudocode}
\usepackage{color}
\usepackage{bm}
\usepackage{times}
\usepackage{booktabs}
\usepackage{multirow}

\newcommand{\abs}[1]{\left|#1\right|}
\newcommand{\norm}[1]{\left\|#1\right\|}
\newcommand{\ket}[1]{\left|#1\right\rangle}
\newcommand{\bra}[1]{\left\langle #1\right|}
\newcommand{\braket}[2]{\left\langle #1|#2\right\rangle}
\newcommand{\ketbra}[2]{\ket{#1}\!\!\bra{#2}}

\newcommand{\End}{\operatorname{End}}
\newcommand{\id}{\mathds{1}}
\newcommand{\R}{\mathbb{R}}
\newcommand{\C}{\mathbb{C}}

\newcommand{\cA}{\mathcal{A}}
\newcommand{\cE}{\mathcal{E}}
\newcommand{\SM}{Supplemental Material}
\newcommand{\bal}{\boxtimes}
\newcommand{\phys}{\mathrm{phys}}
\newcommand{\even}{\mathrm{even}}
\newcommand{\odd}{\mathrm{odd}}

\newcommand{\Span}{\operatorname{span}}

\newcommand{\cN}{\mathcal{N}}

\theoremstyle{plain}
\newtheorem{theorem}{Theorem}
\newtheorem{proposition}{Proposition}
\newtheorem{lemma}{Lemma}
\newtheorem{corollary}{Corollary}
\theoremstyle{definition}
\newtheorem{definition}{Definition}
\newtheorem{remark}{Remark}

\makeatletter
\renewcommand\part[1]{%
  \clearpage
  \onecolumngrid
  \section*{#1}
  
}
\makeatother

\makeatletter
\let\origaddcontentsline\addcontentsline
\renewcommand{\addcontentsline}[3]{}
\makeatother

\begin{document}

\part{}

\title{Hidden Complex Structure in Quotient-Space Real Quantum Mechanics}

\author{Jeongho~Bang}\email{jbang@yonsei.ac.kr}
\affiliation{Institute for Convergence Research and Education in Advanced Technology, Yonsei University, Seoul 03722, Republic of Korea}
\affiliation{Department of Quantum Information, Yonsei University, Incheon 21983, Republic of Korea}

\author{Kyoungho~Cho}
\affiliation{Institute for Convergence Research and Education in Advanced Technology, Yonsei University, Seoul 03722, Republic of Korea}
\affiliation{Department of Statistics and Data Science, Yonsei University, Seoul 03722, Republic of Korea}

\author{Kyunghyun~Baek}
\affiliation{Institute for Convergence Research and Education in Advanced Technology, Yonsei University, Seoul 03722, Republic of Korea}
\affiliation{Department of Quantum Information, Yonsei University, Incheon 21983, Republic of Korea}

\date{\today}

\begin{abstract}
Barrios Hita \textit{et al.} [Phys. Rev. Lett. \textbf{136}, 240202 (2026)] argued that quantum mechanics can be formulated over the real numbers by replacing the tensor-product postulate with a quotient-space construction, and concluded that complex numbers are therefore a matter of convenience. We show that the operational content of this construction is not that of a generic real Hilbert-space theory. Empirical equivalence requires a distinguished real linear operator $\hat{J}$ with $\hat{J}^2=-\hat{\id}$, and all physical effects, instruments, and dynamics must preserve the corresponding $SO(2)$ gauge. Moreover, the composite-system rule is a balanced tensor product over this hidden complex structure, not the ordinary tensor product over $\R$. In multipartite network scenarios, this changes the meaning of source independence: canonical real representatives are not source-factorizable in the usual tensor-product sense. Thus, the construction is best understood as standard complex quantum mechanics written in real notation, not as an independent real-amplitude theory. This clarifies what is, and is not, excluded by experiments testing the necessity of complex numbers.
\end{abstract}

\maketitle

{\em Introduction.}---The question whether complex numbers are essential in quantum theory has two logically distinct forms. One may ask whether the symbol $i$ can be removed from the notation by rewriting complex vectors as pairs of real vectors. That representational question has long had affirmative answers in Stueckelberg-type and rebit formulations~\cite{Stueckelberg1960,Myrheim1999,McKague2009,Wootters2012,Aleksandrova2013}. A sharper question is whether the physical structure carried by the complex field---the imaginary unit, the phase gauge, the algebra of observables, and the tensor product over $\C$---can be eliminated rather than hidden~\cite{BangChoYee2026}. Recent Bell-network no-go results and experiments address a particular foil theory: real-amplitude quantum mechanics with the ordinary tensor product over $\R$~\cite{Renou2021,YingEtAl2025,Li2022,Chen2022,Wu2022}. In this theory, the real restriction is operational, not merely notational.

Barrios Hita \textit{et al.}~\cite{BarriosHita2026} proposed a different real formulation. Their single-system map appends a two-dimensional flag representing real and imaginary parts. Their composite rule does not use the ordinary real tensor product as the physical composition; it quotients the tensor product by equivalence relations encoding the freedom to distribute a global complex phase among subsystems. Since this quotient construction is one-to-one with complex quantum mechanics, it reproduces all complex predictions, including the Renou \textit{et al.} Bell-network value~\cite{Renou2021}. The authors therefore conclude that real-valued quantum mechanics cannot be falsified and that complex numbers are useful but unnecessary.

However, we identify the nontrivial physical structures required for its equivalence to complex quantum mechanics. Our conclusion is that the construction removes complex numbers only syntactically. The imaginary unit reappears as a selected real operator $\hat{J}$, the physical algebra is restricted to the commutant of $\hat{J}$, dynamics must respect the same gauge, and composite systems are formed with the balanced tensor product over that hidden complex structure. Consequently, the quotient-space theory is not generic real-amplitude quantum mechanics. It is complex quantum mechanics expressed as a real gauge-redundant representation.

The distinction is not semantic. If one starts with a real Hilbert space and allows all real symmetric effects, then the flag that stores the imaginary part becomes measurable. If one forms composites by the ordinary real tensor product, then the phase of a product state cannot be reassigned between factors in the same way as in complex quantum mechanics. Conversely, if one forbids the flag measurements and quotients the real tensor product by precisely the phase-balancing relations, complex quantum mechanics is recovered. The question is therefore where the complex structure has gone. Our answer is that it has moved from the scalar field into the operational rules.

The Letter is organized around three claims. First, the single-system flag enforces a superselection rule: not every real observable on the enlarged space is physical. Second, the multipartite quotient is a balanced tensor product, hence a real presentation of $\otimes_{\C}$ rather than a new real composition law. Third, in network experiments and open dynamics the same distinction changes the operational assumptions, most visibly the meaning of source independence and the class of admissible real channels.

{\em Single systems: the flag is a superselection sector.}---Let $H$ be a $d$-dimensional complex Hilbert space and let $V=H_{\R}$ denote its realification. The multiplication by $i$ is the real linear map
\begin{eqnarray}
\hat{J}:V \to V, \quad \hat{J}^2=-\hat{\id}.
\label{eq:J_def_main}
\end{eqnarray}
In the flag notation of Ref.~\cite{BarriosHita2026}, a vector $\ket{\psi}\in\C^d$ is represented by
\begin{eqnarray}
S\ket{\psi} = \operatorname{Re}\ket{\psi} \otimes \ket{0}_F+
\operatorname{Im}\ket{\psi}\otimes\ket{1}_F,
\label{eq:S_main}
\end{eqnarray}
so that multiplication by $i$ is implemented by the real antisymmetric flag operator
\begin{eqnarray}
\hat{J}_F = \begin{pmatrix} 0 & -1 \\ 1 & 0 \end{pmatrix},
\quad
e^{i\alpha} \longleftrightarrow e^{\alpha \hat{J}_F} \in SO(2).
\label{eq:JF_main}
\end{eqnarray}
Thus, $\ket{\psi}$ and $i\ket{\psi}$, although physically identical in complex quantum mechanics, may be orthogonal real vectors after the map $S$. This is not a harmless redundancy unless the allowed measurements are correspondingly restricted.

\begin{theorem}[$\hat{J}$-superselection]
\label{thm:J_super_main}
Let $V$ be a real Hilbert space equipped with a complex structure $\hat{J}^2 = -\hat{\id}$, with $\hat{J}^T = -\hat{J}$. Suppose that all vectors on the orbit $v \sim e^{\alpha \hat{J}}v$ are required to be operationally indistinguishable. Then, every real symmetric physical effect $\hat{E}$ must satisfy
\begin{eqnarray}
[\hat{E}, \hat{J}]=0.
\label{eq:commutant_main}
\end{eqnarray}
Conversely, effects satisfying Eq.~(\ref{eq:commutant_main}) are invariant on every $SO(2)$ gauge orbit.
\end{theorem}

For an effect $\hat{E}$, the invariance of probabilities is the condition $v^T \hat{E} v=(e^{\alpha \hat{J}}v)^T \hat{E}(e^{\alpha \hat{J}}v)$ for all $v$ and $\alpha$. Equivalently, $e^{-\alpha \hat{J}}\hat{E}e^{\alpha \hat{J}}=\hat{E}$. Differentiating at $\alpha=0$ gives $[\hat{E},\hat{J}]=0$, and the converse follows by exponentiation. Hence, a real effect that does not commute with $\hat{J}$ would make the phase gauge observable. In the flag model, a flag-only measurement such as $\hat{\id} \otimes \ketbra{0}{0}_F$ is therefore forbidden. The allowed representatives of complex effects are precisely of the form
\begin{eqnarray}
T(\hat{A}) = \operatorname{Re}(\hat{A}) \otimes \hat{I}_F + \operatorname{Im}(\hat{A}) \otimes \hat{J}_F,
\label{eq:T_main}
\end{eqnarray}
which is the real commutant of the chosen complex structure. This is a superselection rule, not a consequence of using real numbers alone.

The same observation applies beyond projective measurements. If a real operational theory treats the flag as an ordinary degree of freedom, then generic flag couplings, flag dephasing, or flag readout are allowed and standard complex quantum mechanics is not recovered. Empirical equivalence requires that all physical effects, instruments, and transformations preserve the $SO(2)$ gauge generated by $\hat{J}$. The flag is therefore not an accessible subsystem; it is a gauge sector carrying the missing imaginary unit.

This conclusion is stronger than the familiar observation that a complex vector can be stored as two real vectors. It says that the real theory must choose a particular operator $\hat{J}$ and make it unobservable as a gauge generator. Without such a choice, $\R^{2d}$ is only a real vector space with many possible complex structures and many real measurements that have no complex counterpart. With such a choice, the physical algebra is exactly the algebra of complex-linear Hermitian operators written in real blocks. Thus, the single-system theory is not ``real Hilbert space plus Born rule''; it is real Hilbert space plus a protected complex structure.

The superselection rule also shows why the word ``flag'' can be misleading. A flag in ordinary quantum information is a system one may condition on, couple to an apparatus, or trace out after a measurement. The present flag cannot be used in that way. It is more analogous to a gauge coordinate: different flag angles represent the same physical ray, and the theory is empirically equivalent to complex quantum mechanics only when no allowed operation can reveal that angle. The statement in Ref.~\cite{BarriosHita2026} that the flag is not directly accessible is therefore not an innocuous technical remark; it is one of the central physical postulates of the construction.

Equivalently, the observable algebra is not $\mathrm{Sym}(2d, \R)$ but the much smaller algebra $\mathrm{Sym}(2d,\R)\cap\{\hat{J}\}'$. This restriction is experimentally invisible only because it is imposed at the outset. If it is relaxed even slightly, the theory predicts measurements that distinguish different representatives of the same complex ray. Thus, the real formalism has two layers: the kinematic layer of real vectors and the physical layer selected by $\hat{J}$. The latter is where complex quantum mechanics is encoded.

{\em Composite systems: the quotient is a balanced tensor product.}---The central step of Ref.~\cite{BarriosHita2026} is the replacement of the ordinary tensor product over $\R$ by a quotient space. This can be stated invariantly. Let $(V_A, \hat{J}_A)$ and $(V_B, \hat{J}_B)$ be real Hilbert spaces equipped with complex structures. Define
\begin{eqnarray}
V_A \bal V_B &=& (V_A \otimes_{\R} V_B)/\mathcal N_{AB}, \nonumber \\
\mathcal N_{AB} &=& \operatorname{span}_{\R} \{\hat{J}_A v\otimes w-v\otimes \hat{J}_B w\}.
\label{eq:balanced_main}
\end{eqnarray}
The relation in Eq.~(\ref{eq:balanced_main}) is the balanced relation for scalar multiplication by $i$: applying $i$ to the first factor is identified with applying $i$ to the second factor. It immediately implies $\hat{J}_A v \otimes \hat{J}_B w \sim -v \otimes w$, exactly the real form of phase redistribution.

\begin{theorem}[Quotient composition is complex composition]
\label{thm:balanced_main}
Let $H_A$ and $H_B$ be the complex Hilbert spaces whose realifications are $(V_A, \hat{J}_A)$ and $(V_B, \hat{J}_B)$. Then,
\begin{eqnarray}
V_A \bal V_B \simeq (H_A \otimes_{\C} H_B)_{\R}
\label{eq:balanced_iso_main}
\end{eqnarray}
canonically as real Hilbert spaces with complex structure.
\end{theorem}

The isomorphism sends $v \bal w$, the equivalence class of $v\otimes w$, to $v \otimes_{\C}w$. It is well defined exactly because $\hat{J}_A v \otimes w$ and $v \otimes \hat{J}_B w$ represent the same complex tensor. It is surjective by construction and has the correct real dimension, $2 d_A d_B$. Thus, the quotient is not an alternative composition over the real field. It is the complex tensor product written as a real balanced tensor product~\footnote{Further details, including a coordinate-free proof and the equivalence with the flag-kernel quotient, are given in Sec.~IV of the \SM.}.

In the notation of Ref.~\cite{BarriosHita2026}, this is exactly the statement that the kernel of $(S^{-1})^{\otimes2}$ is generated, on the two local flags, by $\ket{00}_F + \ket{11}_F$ and $\ket{01}_F - \ket{10}_F$. These are not arbitrary identifications. They are the real-coordinate form of $(i v) \otimes w = v \otimes (i w)$ and $(iv) \otimes (iw) = -v \otimes w$. Hence, the quotient encodes the rule that the factor $i$ may be moved from one subsystem to another, which is precisely the scalar-balancing rule of the tensor product over $\C$.

This also clarifies the dimension count. The ordinary tensor product $V_A \otimes_{\R} V_B$ has real dimension $4 d_A d_B$, twice the real dimension of the complex composite. The quotient removes exactly the redundant half, leaving $2 d_A d_B$. The removal is not an empirical consequence of locality; it is the imposition of the complex scalar relation in real language. Once that relation is imposed, the resulting theory must agree with complex quantum mechanics, because the composite state space is isomorphic to the realification of the complex composite state space.

This point is decisive for the interpretation of Bell-network tests. The no-go theorems constrain theories in which independent real systems compose with $\otimes_{\R}$. The quotient construction evades those theorems by changing the composition rule to $\bal$. That change may be mathematically consistent, but it is an additional composition postulate. It is not obtained from real Hilbert spaces alone.

The relation also explains why the naive map $S^{\otimes2}$ fails. In complex quantum mechanics, $\psi_A\otimes\psi_B$ and $(i\psi_A)\otimes(-i\psi_B)$ are the same vector. Applying the single-system real map to each factor separately sends them to different real flag vectors unless the quotient identifies those flag vectors. The quotient therefore repairs the map only by imposing the complex balancing relation that the ordinary real tensor product lacks.

{\em The locality postulate is underdetermined.}---Ref.~\cite{BarriosHita2026} motivates its construction by replacing the tensor-product postulate with a locality condition, (P4): operations on one subsystem act trivially on the others, and local operators commute. In formal terms, however, such a condition is imposed only after one has specified an independent-preparation embedding $\xi_F : H_1 \times H_2 \to H$ and local-operation embeddings $\eta_F^i : \End(H_i) \to \End(H)$. (P4) constrains these embeddings; it does not determine them.

Indeed, the ordinary tensor product satisfies (P4) with $\xi(v,w)=v\otimes w$ and $\eta^1(\hat{A})=\hat{A}\otimes\hat{\id}$, $\eta^2(\hat{B})=\hat{\id}\otimes\hat{B}$. The quotient construction also satisfies (P4), but with different embeddings, namely $\xi_R(v,w) = v \bal w$ and $\eta_R^i$ defined through the quotient. Therefore,
\begin{eqnarray}
\mathrm{(P4)} + \mathrm{real~Hilbert~spaces} \not\Rightarrow \bal.
\label{eq:P4_underdet_main}
\end{eqnarray}
The physical content has merely been shifted: instead of postulating $\otimes_{\R}$, one postulates the quotient embedding $\bal$. Since $\bal$ is isomorphic to $\otimes_{\C}$, the recovery of complex predictions follows from the chosen embedding, not from locality alone.

One can phrase the issue as an underdetermination theorem. A condition of the form ``local operations act trivially on remote systems'' is a compatibility condition between a state-composition map and an operator-embedding map. It is satisfied by the standard complex tensor product, by the ordinary real tensor product, and by the balanced quotient tensor product, provided the corresponding embeddings are chosen. Therefore, (P4) cannot by itself select among these theories. The nontrivial physical choice is not (P4) alone, but the selection of the quotient embeddings that make the diagrams with the complex theory commute.

This matters because the conclusion of Ref.~\cite{BarriosHita2026} is sometimes phrased as if a more physical locality postulate replaces the mathematical tensor-product postulate. The formal situation is different. The ordinary tensor product has been replaced by another mathematical composition rule, $\bal$, whose physical meaning is exactly the permission to move the hidden $\hat{J}$ factor between subsystems. That rule is natural if one already wants the real theory to be a realification of complex quantum mechanics. It is not forced by local commutativity.

This is the precise point at which a possible falsification claim changes target. The experiments that rule out the ordinary real tensor product do not rule out a quotient theory isomorphic to complex quantum mechanics, but neither does the quotient theory show that those experiments were testing an irrelevant assumption. They were testing a well-defined alternative operational theory. The quotient theory is a different alternative, with a different composition postulate.

{\em Bell networks and source independence.}---The distinction is operationally visible in the entanglement-swapping network of Renou \textit{et al.}~\cite{Renou2021}. In complex quantum mechanics, two independent sources prepare
\begin{eqnarray}
\hat{\rho}_{AB_1} \otimes_{\C} \hat{\rho}_{B_2C}.
\label{eq:complex_sources_main}
\end{eqnarray}
In the quotient real formulation, the corresponding expression is
\begin{eqnarray}
T_1(\hat{\rho}_{AB_1})\bal T_1(\hat{\rho}_{B_2C}),
\label{eq:real_sources_main}
\end{eqnarray}
where $T_1$ is the trace-normalized real map for density operators. If one writes a canonical representative inside the ambient tensor product of local real flags, the two Bell sources are accompanied not by independent source-local flags but by a global parity flag. For the two Bell-pair state one obtains, in the notation of Ref.~\cite{BarriosHita2026},
\begin{eqnarray}
\ket{\phi^+}_{AB_1}\ket{\phi^+}_{B_2C} \otimes \ket{\psi^{(4)}_{\even}}_{A' B_1' B_2' C'}.
\label{eq:global_flag_main}
\end{eqnarray}
The four-party flag state is an even-parity canonical representative; it is not factorized according to the two independent sources in the ordinary tensor-product sense~\footnote{The general $N$-party canonical even/odd flag states used here are defined in Sec.~V of the \SM.}.

The quotient theory reconciles Eq.~(\ref{eq:real_sources_main}) with Eq.~(\ref{eq:global_flag_main}) by declaring the relevant factorization to be quotient factorization. But this is precisely the change of assumption. If the flags are physical local ancillas, the canonical representative contains a global parity constraint across the network. If the flags are gauge variables, the construction is a real encoding of the complex theory. In neither case is it the real-amplitude tensor-product theory excluded by Bell-network experiments. 

For the Renou network, this distinction is not merely formal. The complex construction uses two Bell-state sources, Bob's Bell measurement, and Alice--Charlie measurements involving $\hat{X}, \hat{Y}, \hat{Z}$ directions; it gives the value $T = 6\sqrt2 \simeq 8.4853$. The real-amplitude tensor-product model has the lower bound $T \le 7.6605$. The quotient real model reaches the complex value because all states and measurements are mapped through $T$ and composed through $\bal$. Thus, the numerical violation is not a new real-amplitude resource. It is the complex resource transported through the quotient isomorphism.

The source-independence assumption is therefore doing different work in the two comparisons. In the real-amplitude no-go theorem, the independence means ordinary tensor-product factorization of the two source states. In the quotient theory, the independence means factorization under $\bal$. When one chooses canonical representatives in the ambient flag tensor product, $\bal$-factorized states can appear nonfactorized. The global even-parity flag in Eq.~(\ref{eq:global_flag_main}) is the concrete witness of this change.

The experiments remain still meaningful: they rule out ordinary real-amplitude quantum mechanics as a foil theory, not a gauge-redundant real rewriting of complex quantum mechanics.

{\em Dynamics, open systems, and overhead.}---The same superselection requirement applies to open dynamics. Let $\mathcal A_{\hat{J}} = \{ \hat{A}:[\hat{A},\hat{J}]=0 \}$ be the physical algebra. A generic real completely positive trace-preserving map on the enlarged real space need not preserve $\mathcal{A}_{\hat{J}}$ and need not map the gauge-equivalent states to the gauge-equivalent states. A sufficient and natural condition is $\hat{J}$-covariance,
\begin{eqnarray}
\cE(e^{\alpha \hat{J}}\hat{\rho} e^{-\alpha \hat{J}}) = e^{\alpha \hat{J}}\cE(\hat{\rho})e^{-\alpha \hat{J}},
\label{eq:J_cov_main}
\end{eqnarray}
or, in the Heisenberg picture,
\begin{eqnarray}
\cE^\dagger(\mathcal A_{\hat{J}})\subseteq \mathcal A_{\hat{J}}.
\label{eq:J_heis_main}
\end{eqnarray}
The realifications of complex channels satisfy this condition, but the generic real environment couplings do not. Thus, the price of equivalence is not merely the addition of a flag in kinematics; it is a protection rule for all admissible noise, measurements, and instruments.

It would be misleading, however, to claim that the abstract quotient state space has an exponential dimension overhead. For $N$ $d$-level parties, the real dimension of $(\C^d)^{\otimes N}$ is $2d^N$, and the quotient has the same dimension. The overhead appears only if the canonical representatives are interpreted as physical vectors in the ambient local-flag space $(\R^2)^{\otimes N}$: the even and odd flag states occupy a $2^N$-dimensional ambient flag space and impose global parity structure. The theory must therefore choose between a quotient/gauge interpretation with no extra physical flags, and a literal local-ancilla interpretation with nonlocal flag constraints.

This observation is especially relevant for open quantum systems. A real heat bath or measurement device that can couple to the flag basis would implement a perfectly ordinary real channel, but it would not correspond to any complex channel. Thus, the quotient theory must assume that the $\hat{J}$-generated gauge  is protected against all physical interactions. In finite dimension this assumption is enough to realify complex CPTP maps without adding an infinite reservoir of auxiliary degrees of freedom. The price is instead a restriction on the allowed real channels: they must be those that descend to the quotient or, more strongly, those covariant under the $SO(2)$ gauge.

This also answers a natural concern about the resources. The decisive issue is not that a finite-dimensional quotient formulation necessarily needs exponentially many independent physical degrees of freedom. Abstractly it does not. The decisive issue is that the equivalence class, or equivalently, the allowed-operation rule, must be maintained for arbitrarily complicated dynamics. For closed systems, this means the orthogonal evolutions commuting with the hidden complex structure. For open systems, it means the channels that respect the same gauge. The generic real dynamics would be a different theory.

{\em Discussion.}---The quotient-space formulation of Ref.~\cite{BarriosHita2026} is a consistent real representation of complex quantum mechanics. What it does not show is that the complex structure has been physically eliminated. The empirical equivalence relies on three ingredients that are absent from generic real-amplitude quantum mechanics: a distinguished complex structure $\hat{J}$, a superselection rule selecting the $\hat{J}$-commutant as the physical algebra, and a balanced tensor product over $\hat{J}$ for composite systems. These ingredients are exactly what make a complex Hilbert-space theory complex in operational terms.

The status of the uniqueness claims in the target construction should be read accordingly. Under the assumptions that already encode the phase group as an $SO(2)$ representation and require operators to respect that representation, the real flag map and its quotient composition are essentially unique. This is a useful representation theorem. However, it does not show that the complex structure is unnecessary in the physical theory; it shows that once the complex phase structure is reintroduced as gauge data, the real representation is fixed.

This distinction separates two claims that are often conflated. The representational claim is true: complex amplitudes can be encoded by real vectors. The stronger physical claim is not established: the encoded theory still contains the imaginary unit as $\hat{J}$, still restricts observables and dynamics by a phase-gauge rule, and still composes systems by the complex balanced tensor product. Thus, the construction is best understood as standard complex quantum mechanics in real notation, not as an independent real-number theory. A compact way to state the result is therefore: ``real coordinates'' are eliminable, but ``complex structure'' is not eliminated by the quotient construction. The former is a statement about notation; the latter is a statement about the algebra of physical operations and the composition of independent systems.

The complex numbers can be hidden in real notation, but not eliminated from the physical structure. Equivalently, the imaginary unit can be removed from the notation only by reinstating it as a real operator $\hat{J}$, restricting the physical algebra to its commutant, and composing systems with the corresponding balanced tensor product.

{\em Note.}---The \SM~gives full proofs, multipartite flag formulas, and the comparison with earlier real-Hilbert-space formulations. A detailed comparison with recent and earlier works is also given in Sec.~IX of the \SM~(also, refer to Refs.~\cite{YingEtAl2025,HoffreumonWoods2025,HoffreumonWoods2026,KalardeXuRenou2026,FengRenVedral2025,WeilenmannGisinSekatski2025,Volovich2025,MaioliCuradoGazeau2026,MorettiOppio2017}).

\begin{acknowledgments}
{Acknowledgments.}---This work was supported by the Ministry of Science, ICT and Future Planning (MSIP) by the National Research Foundation of Korea (RS-2024-00432214); the Institute of Information and Communications Technology Planning and Evaluation grant funded by the Korean government (RS-2019-II190003, ``Research and Development of Core Technologies for Programming, Running, Implementing and Validating of Fault-Tolerant Quantum Computing System''); the Ministry of Trade, Industry and Resources (MOTIR), Korea, under the project ``Industrial Technology Infrastructure Program'' (RS-2024-00466693); the Korean ARPA-H Project through the Korea Health Industry Development Institute (KHIDI), funded by the Ministry of Health \& Welfare, Republic of Korea (RS-2025-25456722). We acknowledge the Yonsei University Quantum Computing Project Group for providing support and access to the Quantum System One (Eagle Processor), which is operated at Yonsei University.
\end{acknowledgments}

%


\clearpage
\newpage
\onecolumngrid

\makeatletter
\let\addcontentsline\origaddcontentsline
\makeatother

\part{Supplemental Material for ``Hidden Complex Structure in Quotient-Space Real Quantum Mechanics''}


\section{Purpose and notation}

This Supplemental Material gives the detailed mathematical statements behind the Letter. The main point is not that the quotient-space construction of Ref.~\cite{SM-BarriosHita2026} is inconsistent. On the contrary, precisely because it is naturally isomorphic to standard complex quantum mechanics, it is internally consistent in finite dimension. The point is that the equivalence requires additional structure that is not present in a generic real-amplitude Hilbert-space theory: a distinguished complex structure $\hat{J}$, a superselection rule selecting the $\hat{J}$-commutant as the physical algebra, and a balanced tensor product over that structure.

We use the following notation. If $H$ is a complex Hilbert space, $H_{\R}$ denotes the same additive group regarded as a real vector space. Multiplication by $i$ is denoted by $\hat{J}$; thus $\hat{J}^2=-\hat{\id}$. We write $\End_{\R}(V)$ for real linear maps on a real vector space $V$, and $\End_{\C}(H)$ for complex linear maps. The symbol ``$\bal$'' denotes the balanced tensor product introduced below. It is the invariant version of the flag quotient tensor product denoted ``$\otimes_F$'' in Ref.~\cite{SM-BarriosHita2026}.


\section{Realification and complex structures}\label{sec:realification}

Let $H\simeq\C^d$ be a finite-dimensional complex Hilbert space with inner product linear in the second argument. The realification $V=H_{\R}$ is the $2d$-dimensional real vector space obtained by restricting scalar multiplication from $\C$ to $\R$. In a basis, a vector $\psi=x+iy$ is identified with $(x,y)\in\R^d\oplus\R^d$. Multiplication by $i$ becomes
\begin{eqnarray}
\hat{J}(x,y)=(-y,x),
\quad
\hat{J}=\begin{pmatrix}0&-\hat{\id}_d\\ \hat{\id}_d&0\end{pmatrix},
\quad
\hat{J}^2=-\hat{\id}_{2d}.
\label{eq:sm_J_block}
\end{eqnarray}
The real inner product is
\begin{eqnarray}
(v,w)_{\R}=\operatorname{Re}\braket{v}{w}_{\C},
\label{eq:sm_real_inner}
\end{eqnarray}
and $\hat{J}$ is orthogonal and antisymmetric, $\hat{J}^T\hat{J}=\hat{\id}$ and $\hat{J}^T=-\hat{J}$.

The flag representation used in Ref.~\cite{SM-BarriosHita2026} is the same construction written as
\begin{eqnarray}
S:\C^d\to\R^d\otimes\R^2_F,
\quad
S\ket{\psi}=\operatorname{Re}\ket{\psi}\otimes\ket{0}_F+
\operatorname{Im}\ket{\psi}\otimes\ket{1}_F.
\label{eq:sm_S_def}
\end{eqnarray}
In this notation,
\begin{eqnarray}
\hat{J}_F\ket{0}_F=\ket{1}_F,
\quad
\hat{J}_F\ket{1}_F=-\ket{0}_F,
\quad
\hat{J}_F=\begin{pmatrix}0&-1\\1&0\end{pmatrix}.
\label{eq:sm_JF}
\end{eqnarray}
Thus, the $U(1)$ phase orbit of a complex ray becomes an $SO(2)$ orbit in the real flag plane,
\begin{eqnarray}
S(e^{i\alpha}\psi)=e^{\alpha \hat{J}}S(\psi).
\label{eq:sm_phase_orbit}
\end{eqnarray}

\begin{proposition}[Complex-linear maps are the $\hat{J}$-commutant]\label{prop:commutant}
Let $V=H_{\R}$. A real linear map $A_R\in\End_{\R}(V)$ is the realification of a complex-linear map $A\in\End_{\C}(H)$ if and only if
\begin{eqnarray}
[A_R,\hat{J}]=0.
\label{eq:sm_commutant_end}
\end{eqnarray}
In the block form $V\simeq\R^d\oplus\R^d$, such maps have the form
\begin{eqnarray}
A_R=\begin{pmatrix}X&-Y\\Y&X\end{pmatrix},
\label{eq:sm_complex_block}
\end{eqnarray}
which represents the complex matrix $A=X+iY$.
\end{proposition}

\begin{proof}---If $A$ is complex linear, then $A(i\psi)=iA\psi$, which is exactly $A_R\hat{J}=\hat{J}A_R$. Conversely, let $A_R$ commute with $\hat{J}$. Write $A_R$ in real blocks,
\begin{eqnarray}
A_R=\begin{pmatrix}P&Q\\R&S\end{pmatrix}.
\end{eqnarray}
The equation $A_R\hat{J}=\hat{J}A_R$ gives $Q=-R$ and $S=P$. Hence, $A_R$ has the form Eq.~(\ref{eq:sm_complex_block}) with $X=P$ and $Y=R$, which is multiplication by $X+iY$ on $x+iy$.
\end{proof}

For Hermitian operators, the corresponding real matrices are real symmetric and commute with $\hat{J}$:
\begin{eqnarray}
T(\hat{A})=\operatorname{Re}(\hat{A})\otimes \hat{I}_F+\operatorname{Im}(\hat{A})\otimes \hat{J}_F.
\label{eq:sm_T_def}
\end{eqnarray}
When $\hat{A}=\hat{A}^\dagger$, $\operatorname{Re}(\hat{A})^T=\operatorname{Re}(\hat{A})$ and $\operatorname{Im}(\hat{A})^T=-\operatorname{Im}(\hat{A})$; since $\hat{J}_F^T=-\hat{J}_F$, Eq.~(\ref{eq:sm_T_def}) is real symmetric. This proves that the observable algebra of complex quantum mechanics is not the full algebra of real symmetric matrices on $V$. It is the real symmetric part of the $\hat{J}$-commutant.

\begin{remark}[Conditional uniqueness]
The uniqueness result in Ref.~\cite{SM-BarriosHita2026} assumes, among other conditions, real linearity, preservation of length and orthogonality, and the requirement that the complex $U(1)$ phase ambiguity be represented by an $SO(2)$ action. Under those assumptions the flag representation is natural and essentially unique. This is a conditional uniqueness theorem: once one decides to encode the phase group as an $SO(2)$ gauge and to keep only gauge-preserving observables, the representation is fixed up to isomorphism. It is not a derivation of a real theory without complex structure.
\end{remark}

\section{$\hat{J}$-superselection and forbidden flag measurements}\label{sec:superselection}

The physical state of a nonzero vector in complex quantum mechanics is a ray. In the real flag representation, the corresponding redundancy is not merely a sign but the full orbit generated by $\hat{J}$:
\begin{eqnarray}
v\sim e^{\alpha \hat{J}}v, \quad \alpha\in\R.
\label{eq:sm_gauge_equiv_vec}
\end{eqnarray}
Operational equivalence requires that every allowed effect assign the same probability to all representatives of the same orbit.

\begin{theorem}[Superselection of the complex structure]\label{thm:sm_superselection}
Let $V$ be a finite-dimensional real Hilbert space equipped with $\hat{J}^2=-\hat{\id}$ and $\hat{J}^T=-\hat{J}$. Let $\hat{E}=\hat{E}^T$ be a real effect. The following are equivalent:
\begin{enumerate}
\item[(i)] $v^T \hat{E} v=(e^{\alpha \hat{J}}v)^T \hat{E}(e^{\alpha \hat{J}}v)$ for all $v\in V$ and all $\alpha\in\R$.
\item[(ii)] $e^{-\alpha \hat{J}}\hat{E}e^{\alpha \hat{J}}=\hat{E}$ for all $\alpha\in\R$.
\item[(iii)] $[\hat{E},\hat{J}]=0$.
\end{enumerate}
\end{theorem}

\begin{proof}---The equivalence of (i) and (ii) follows because a real symmetric quadratic form that vanishes on all $v$ is the zero form. Since $(e^{\alpha \hat{J}})^T=e^{-\alpha \hat{J}}$, condition (i) is equivalent to
\begin{eqnarray}
v^T\bigl(e^{-\alpha \hat{J}}\hat{E}e^{\alpha \hat{J}}-\hat{E}\bigr)v=0
\end{eqnarray}
for all $v$. This implies (ii). Differentiating (ii) at $\alpha=0$ gives
\begin{eqnarray}
-\hat{J}\hat{E}+\hat{E}\hat{J}=0,
\end{eqnarray}
which is (iii). Conversely, if $[\hat{E},\hat{J}]=0$, then $\hat{E}$ commutes with $e^{\alpha \hat{J}}$ for all $\alpha$, so (ii) holds.
\end{proof}

\begin{corollary}[Flag-only measurements are not physical]
The projectors
\begin{eqnarray}
\hat{E}_0=\hat{\id}_d\otimes\ketbra{0}{0}_F,
\quad
\hat{E}_1=\hat{\id}_d\otimes\ketbra{1}{1}_F
\label{eq:sm_flag_meas}
\end{eqnarray}
are not allowed physical effects in the real representation of complex quantum mechanics.
\end{corollary}

\begin{proof}---Using Eq.~(\ref{eq:sm_JF}), $[\hat{E}_0,\hat{J}_F]\ne0$. Equivalently, take a nonzero real vector $x\in\R^d$. The representatives $x\otimes\ket{0}_F$ and $\hat{J}(x\otimes\ket{0}_F)=x\otimes\ket{1}_F$ correspond to complex states differing by a phase $i$, but
\begin{eqnarray}
(x\otimes\ket{0})^T \hat{E}_0(x\otimes\ket{0})=\norm{x}^2,
\quad
(x\otimes\ket{1})^T \hat{E}_0(x\otimes\ket{1})=0.
\end{eqnarray}
Thus, $\hat{E}_0$ distinguishes representatives of the same complex ray.
\end{proof}

This is the first nontrivial physical assumption of the quotient-space real formulation. If the real flag is an ordinary subsystem, then the above measurement is as legitimate as any other real projector. To reproduce complex quantum mechanics, one must impose a superselection rule forbidding it and every other non-$\hat{J}$-commuting effect. The statement extends to POVMs by applying the theorem to every effect in the POVM.

The same restriction applies to instruments. Let $\{\mathcal I_r\}$ be an instrument on real density matrices. It is compatible with the gauge if, whenever $\hat{\rho}$ and $e^{\alpha \hat{J}}\hat{\rho} e^{-\alpha \hat{J}}$ are gauge-equivalent, the outcome probabilities and posterior states are gauge-equivalent. A sufficient condition is covariance of every completely positive branch,
\begin{eqnarray}
\mathcal I_r(e^{\alpha \hat{J}}\hat{\rho} e^{-\alpha \hat{J}})=e^{\alpha \hat{J}}\mathcal I_r(\hat{\rho})e^{-\alpha \hat{J}}.
\label{eq:sm_instrument_cov}
\end{eqnarray}
Without such a condition, the instrument can leak information about the flag phase and thereby destroy equivalence with complex quantum mechanics.

\section{Quotient tensor product as a balanced tensor product}\label{sec:balanced}

The key composite-system construction in Ref.~\cite{SM-BarriosHita2026} is obtained by taking an ordinary real tensor product and quotienting by the kernel of the tensor product of inverse single-system maps. We now rewrite this construction in coordinate-free form.

\begin{definition}[Balanced real tensor product over $\hat{J}$]\label{def:balanced}
Let $(V_A,\hat{J}_A)$ and $(V_B,\hat{J}_B)$ be real vector spaces with complex structures. Define
\begin{eqnarray}
V_A\bal V_B
=\frac{V_A\otimes_{\R}V_B}{\cN_{AB}},
\quad
\cN_{AB}=\Span_{\R}\{\hat{J}_A v\otimes w-v\otimes \hat{J}_B w: v\in V_A,w\in V_B\}.
\label{eq:sm_balanced_def}
\end{eqnarray}
The equivalence class of $v\otimes w$ is denoted $v\bal w$.
\end{definition}

The defining relation is
\begin{eqnarray}
\hat{J}_A v\bal w=v\bal \hat{J}_B w.
\label{eq:sm_i_balance}
\end{eqnarray}
Applying it twice gives
\begin{eqnarray}
\hat{J}_A v\bal \hat{J}_B w=v\bal \hat{J}_B^2 w=-v\bal w.
\label{eq:sm_minus_relation}
\end{eqnarray}
Thus, the product of two local $i$ factors is identified with the scalar $-1$, as in complex tensor products.

\begin{theorem}[Balanced quotient equals complex tensor product]\label{thm:sm_balanced_iso}
Let $H_A$ and $H_B$ be complex Hilbert spaces, and let $(V_A,\hat{J}_A)=(H_A)_{\R}$ and $(V_B,\hat{J}_B)=(H_B)_{\R}$. Then, the map
\begin{eqnarray}
\Phi:V_A\bal V_B\to (H_A\otimes_{\C}H_B)_{\R},
\quad
\Phi(v\bal w)=v\otimes_{\C}w
\label{eq:sm_Phi_def}
\end{eqnarray}
is a canonical real-linear isomorphism.
\end{theorem}

\begin{proof}---The map from $V_A\times V_B$ to $(H_A\otimes_{\C}H_B)_{\R}$ given by $(v,w)\mapsto v\otimes_{\C}w$ is real bilinear. Moreover,
\begin{eqnarray}
\Phi(\hat{J}_A v\otimes w-v\otimes \hat{J}_B w)=iv\otimes_{\C}w-v\otimes_{\C}iw=0,
\end{eqnarray}
so it vanishes on $\cN_{AB}$ and descends to a well-defined map on the quotient. The image contains all simple complex tensors and hence spans $H_A\otimes_{\C}H_B$ over $\C$, therefore over $\R$. Finally,
\begin{eqnarray}
\dim_{\R}(V_A\bal V_B)=2\dim_{\C}H_A\dim_{\C}H_B
=\dim_{\R}(H_A\otimes_{\C}H_B)_{\R},
\end{eqnarray}
so the surjective map is an isomorphism. The dimension formula follows by choosing complex bases $\{e_a\}$ and $\{f_b\}$ and observing that each pair contributes two real classes, represented by $e_a\bal f_b$ and $\hat{J}_A e_a\bal f_b$.
\end{proof}

\begin{proposition}[Equivalence with the flag-kernel quotient]\label{prop:kernel_equiv}
For two single-system flag spaces, the balanced relation in Eq.~(\ref{eq:sm_balanced_def}) is equivalent to the kernel relations used in Ref.~\cite{SM-BarriosHita2026}, namely
\begin{eqnarray}
\ket{00}_{F_AF_B}+\ket{11}_{F_AF_B}\sim0,
\quad
\ket{01}_{F_AF_B}-\ket{10}_{F_AF_B}\sim0.
\label{eq:sm_kernel_relations}
\end{eqnarray}
\end{proposition}

\begin{proof}---Using $\hat{J}\ket{0}=\ket{1}$ and $\hat{J}\ket{1}=-\ket{0}$, the balanced relation $\hat{J}_A v\otimes w-v\otimes \hat{J}_B w\sim0$ gives, for $v=\ket{0}_{F_A}$ and $w=\ket{0}_{F_B}$,
\begin{eqnarray}
\ket{10}-\ket{01}\sim0,
\end{eqnarray}
which is the second relation in Eq.~(\ref{eq:sm_kernel_relations}) up to sign. For $v=\ket{0}_{F_A}$ and $w=\ket{1}_{F_B}$, it gives
\begin{eqnarray}
\ket{11}+\ket{00}\sim0,
\end{eqnarray}
which is the first relation. Conversely, these two relations imply the balanced relations on the flag basis and hence by linearity on the whole tensor product.
\end{proof}

The theorem shows that the quotient composition is not merely similar to the complex tensor product. It is naturally isomorphic to the complex tensor product, regarded as a real vector space, with the complex scalar-balancing relation made explicit. This is why the quotient theory reproduces complex quantum mechanics: the composition rule has been chosen to be isomorphic to complex composition.

For $N$ parties, the construction is associative because tensor products over $\C$ are associative. In real notation, one may define
\begin{eqnarray}
V_1\bal\cdots\bal V_N
=\left(V_1\otimes_{\R}\cdots\otimes_{\R}V_N\right)/\cN_N,
\label{eq:sm_N_balanced}
\end{eqnarray}
where $\cN_N$ is generated by the relations that move $\hat{J}$ from any factor to any other factor. Equivalently,
\begin{eqnarray}
\hat{J}_i v_1\otimes\cdots\otimes v_i\otimes\cdots\otimes v_N
\sim
\hat{J}_j v_1\otimes\cdots\otimes v_j\otimes\cdots\otimes v_N
\label{eq:sm_move_J}
\end{eqnarray}
with $\hat{J}_i$ acting on the $i$th factor. This quotient is canonically isomorphic to $(H_1\otimes_{\C}\cdots\otimes_{\C}H_N)_{\R}$.

\section{Multipartite flag representatives and overhead}\label{sec:flags}

The quotient $V_1\bal\cdots\bal V_N$ has real dimension $2\prod_i d_i$, where $d_i=\dim_{\C}H_i$. Thus, the abstract quotient does not have an exponential overhead relative to complex quantum mechanics. Nevertheless, the canonical representatives used in the ambient local-flag tensor product live inside $(\R^2)^{\otimes N}$ and have a global parity form.

For equal local dimensions, Ref.~\cite{SM-BarriosHita2026} writes the canonical flag representatives as
\begin{eqnarray}
\ket{\psi^{(N)}_{\even}} &=& \frac{1}{\sqrt{2^{N-1}}}\sum_{\substack{k\in\{0,1\}^N\\ |k|_H~\mathrm{even}}}(-1)^{|k|_H/2}\ket{k},
\label{eq:sm_even_flag}\\
\ket{\psi^{(N)}_{\odd}} &=& \frac{1}{\sqrt{2^{N-1}}}\sum_{\substack{k\in\{0,1\}^N\\ |k|_H~\mathrm{odd}}}(-1)^{(|k|_H-1)/2}\ket{k},
\label{eq:sm_odd_flag}
\end{eqnarray}
where $\abs{k}_H$ is the Hamming weight. A complex state $\ket{\Psi}\in(\C^d)^{\otimes N}$ is represented canonically as
\begin{eqnarray}
R(\ket{\Psi})= \operatorname{Re}\ket{\Psi}\otimes\ket{\psi^{(N)}_{\even}} + \operatorname{Im}\ket{\Psi}\otimes\ket{\psi^{(N)}_{\odd}}.
\label{eq:sm_R_N}
\end{eqnarray}

\begin{lemma}[Canonical flags are not preserved by ordinary tensor product]\label{lem:not_preserved}
For positive $N$ and $M$, the ordinary tensor product of canonical flag representatives is not, in general, a linear combination of the two canonical representatives for $N+M$ parties. Instead,
\begin{eqnarray}
\ket{\psi^{(N+M)}_{\even}} &=&\frac{1}{\sqrt2}\left( \ket{\psi^{(N)}_{\even}}\ket{\psi^{(M)}_{\even}} - \ket{\psi^{(N)}_{\odd}}\ket{\psi^{(M)}_{\odd}} \right),
\label{eq:sm_even_composition}\\
\ket{\psi^{(N+M)}_{\odd}} &=& \frac{1}{\sqrt2}\left( \ket{\psi^{(N)}_{\even}}\ket{\psi^{(M)}_{\odd}} + \ket{\psi^{(N)}_{\odd}}\ket{\psi^{(M)}_{\even}} \right).
\label{eq:sm_odd_composition}
\end{eqnarray}
\end{lemma}

\begin{proof}---The parity of the concatenated bit string is the sum of the parities of the two substrings. Even parity arises from even-even and odd-odd, while odd parity arises from even-odd and odd-even. The signs in Eqs.~(\ref{eq:sm_even_flag}) and (\ref{eq:sm_odd_flag}) multiply exactly as shown, with the additional minus sign in the odd-odd contribution to Eq.~(\ref{eq:sm_even_composition}) because $(|k|_H-1)/2+(|l|_H-1)/2=(|kl|_H-2)/2=|kl|_H/2-1$. Normalization gives the factor $1/\sqrt2$.
\end{proof}

Eqs.~(\ref{eq:sm_even_composition}) and (\ref{eq:sm_odd_composition}) are central. From the viewpoint of the ordinary tensor product of local flags, composing two systems changes the canonical flag representative nonlocally. From the quotient viewpoint, this is harmless because the quotient identifies
\begin{eqnarray}
\left[ \ket{\psi^{(N)}_{\even}}\ket{\psi^{(M)}_{\even}} \right] = \left[ -\ket{\psi^{(N)}_{\odd}}\ket{\psi^{(M)}_{\odd}} \right] = \left[ \ket{\psi^{(N+M)}_{\even}} \right],
\label{eq:sm_equiv_even}
\end{eqnarray}
with analogous relations for the odd class. But the distinction matters for physical interpretation. If the local flags are regarded as actual ancillas, their canonical representatives impose global parity correlations across all parties.

\begin{proposition}[No quotient overhead, but ambient flag growth]\label{prop:overhead}
For $N$ parties with local complex dimensions $d_i$, the quotient state space has real dimension $2\prod_i d_i$. However, the canonical representative in Eq.~(\ref{eq:sm_R_N}) is embedded in an ambient flag space of dimension $2^N$.
\end{proposition}

\begin{proof}---The quotient is isomorphic to $(\otimes_i^{\C}H_i)_{\R}$ by the $N$-partite version of Theorem~\ref{thm:sm_balanced_iso}, so its real dimension is $2\prod_i d_i$. On the other hand, Eq.~(\ref{eq:sm_R_N}) uses the two vectors in Eqs.~(\ref{eq:sm_even_flag}) and (\ref{eq:sm_odd_flag}), which are superpositions over bit strings in $(\R^2)^{\otimes N}$, an ambient space of dimension $2^N$.
\end{proof}

Thus, an objection based solely on quotient dimension is not correct: the abstract theory is not larger than complex quantum mechanics. The sharper objection is interpretive. Either the quotient is fundamental, in which case the flags are gauge redundancy and not physical local systems; or one insists on physical local flags, in which case the canonical representatives contain global parity constraints.

\section{Underdetermination of the locality postulate (P4)}\label{sec:P4}

Ref.~\cite{SM-BarriosHita2026} replaces the ordinary tensor-product postulate by a locality postulate (P4): the operations on one subsystem act trivially on the other subsystems, and local operators commute. In this section we spell out why this condition does not determine the quotient tensor product.

\begin{definition}[Independent preparation]
Let $H_1,\ldots,H_N$ be Hilbert spaces over a field $F$. An independent-preparation rule is a multilinear embedding
\begin{eqnarray}
\xi_F:H_1\times\cdots\times H_N \to H,
\quad
(\psi_1,\ldots,\psi_N)\mapsto\xi_F(\psi_1, \ldots, \psi_N),
\label{eq:sm_xi_def}
\end{eqnarray}
where $H$ is the Hilbert space assigned to the composite system.
\end{definition}

\begin{definition}[Local operation embedding]
For a subsystem $i$, a local-operation embedding is a linear map
\begin{eqnarray}
\eta_F^i:\End_F(H_i)\to\End_F(H).
\label{eq:sm_eta_def}
\end{eqnarray}
\end{definition}

\begin{definition}[(P4) in embedding form]
For a bipartite system, (P4) is the requirement that, for all $\hat{A}_1\in\End_F(H_1)$, $\hat{A}_2\in\End_F(H_2)$, and states $\psi_i\in H_i$,
\begin{eqnarray}
\eta_F^1(\hat{A}_1)\xi_F(\psi_1,\psi_2) &=& \xi_F(\hat{A}_1\psi_1,\psi_2),
\label{eq:sm_P4_1}\\
\eta_F^2(\hat{A}_2)\xi_F(\psi_1,\psi_2) &=& \xi_F(\psi_1,\hat{A}_2\psi_2).
\label{eq:sm_P4_2}
\end{eqnarray}
One may additionally require the images of spacelike separated local operations to commute.
\end{definition}

The important point is that Eqs.~(\ref{eq:sm_P4_1}) and (\ref{eq:sm_P4_2}) are constraints on already specified embeddings. They do not select the embeddings.

\begin{proposition}[(P4) does not determine the composite rule]\label{prop:P4_under}
The postulate (P4) is satisfied both by the ordinary real tensor product and by the quotient balanced tensor product. Therefore, (P4) in Ref.~\cite{SM-BarriosHita2026} does not imply the quotient rule.
\end{proposition}

\begin{proof}---For the ordinary real tensor product, take
\begin{eqnarray}
\xi_R^{\otimes}(v,w)=v\otimes_{\R} w,
\quad
\eta_R^{1,\otimes}(\hat{A})=\hat{A}\otimes\hat{\id},
\quad
\eta_R^{2,\otimes}(\hat{B})=\hat{\id}\otimes\hat{B}.
\end{eqnarray}
Then, Eqs.~(\ref{eq:sm_P4_1}) and (\ref{eq:sm_P4_2}) hold by the defining property of the tensor product, and the two local embeddings commute.

For the quotient construction take
\begin{eqnarray}
\xi_R^{\bal}(v,w)=v\bal w,
\label{eq:sm_xi_bal}\\
\eta_R^{1,\bal}(\hat{A})(v\bal w)=(\hat{A}v)\bal w,
\quad
\eta_R^{2,\bal}(\hat{B})(v\bal w)=v\bal (\hat{B}w),
\label{eq:sm_eta_bal}
\end{eqnarray}
for operators compatible with the quotient structure, i.e., operators commuting with the relevant $\hat{J}$'s. Then, (P4) again holds. Since the same locality condition is satisfied by two different composition rules, it cannot determine which one is physical.
\end{proof}

The quotient construction is therefore an additional composition postulate. Its appeal comes from the fact that it reproduces complex quantum mechanics, not from (P4) alone. In slogan form,
\begin{eqnarray}
\mathrm{(P4)} + \mathrm{real~Hilbert~spaces}\not\Rightarrow \otimes_F,
\quad
\otimes_F \simeq \otimes_{\C}\ \mathrm{realified}.
\label{eq:sm_P4_slogan}
\end{eqnarray}
This is the precise sense in which the tensor product postulate is not eliminated, but replaced by a balanced tensor product postulate.

\section{Bell-network source independence}\label{sec:bell}

The Renou \textit{et al.} network~\cite{SM-Renou2021} is a particularly useful place to see the distinction between ordinary real factorization and quotient factorization. The experiment involves two independent sources, one distributing a state to Alice and Bob's first system, and one distributing a state to Bob's second system and Charlie. Bob performs a Bell-state measurement. Alice and Charlie perform measurements chosen from sets that involve the three Pauli directions. The complex model reaches the Bell-network value
\begin{eqnarray}
T=6\sqrt2\simeq 8.4853,
\label{eq:sm_T_complex}
\end{eqnarray}
whereas real-amplitude tensor-product models obey a smaller bound, reported as $T \le 7.6605$ in the analysis of Ref.~\cite{SM-Renou2021} and summarized in Ref.~\cite{SM-BarriosHita2026}.

In complex quantum mechanics, the source-independent state is
\begin{eqnarray}
\hat{\rho}_{AB_1}\otimes_{\C}\hat{\rho}_{B_2C}.
\label{eq:sm_complex_source}
\end{eqnarray}
In the quotient real theory it is written as
\begin{eqnarray}
T_1(\hat{\rho}_{AB_1})\bal T_1(\hat{\rho}_{B_2C}),
\label{eq:sm_quotient_source}
\end{eqnarray}
where
\begin{eqnarray}
T_1(\hat{\rho})=\frac12\operatorname{Re}(\hat{\rho})\otimes \hat{I}+\frac12\operatorname{Im}(\hat{\rho})\otimes \hat{J}
\label{eq:sm_T1_density}
\end{eqnarray}
normalizes the trace in the real representation. This is source-independent with respect to $\bal$, not with respect to the ordinary ambient tensor product.

For the particular two-Bell-pair input,
\begin{eqnarray}
\ket{\phi^+}_{AB_1}\otimes\ket{\phi^+}_{B_2C},
\end{eqnarray}
the canonical real representative is
\begin{eqnarray}
\ket{\tilde\psi_0} = \ket{\phi^+}_{AB_1}\ket{\phi^+}_{B_2C} \otimes \ket{\psi^{(4)}_{\even}}_{A'B_1'B_2'C'}.
\label{eq:sm_global_flag_network}
\end{eqnarray}
Here, the primed systems are the flags. The flag state is not source-factorizable in the ordinary tensor product according to the source split $(A'B_1')|(B_2'C')$.

\begin{proposition}[The four-party even flag is not source-factorizable]\label{prop:flag_not_factor}
Across the bipartition $(A'B_1')|(B_2'C')$, the state $\ket{\psi^{(4)}_{\even}}$ has Schmidt rank $2$. Hence, it is not a product of a flag state for source $AB_1$ and a flag state for source $B_2C$.
\end{proposition}

\begin{proof}---Set $N=M=2$ in Eq.~(\ref{eq:sm_even_composition}). Then,
\begin{eqnarray}
\ket{\psi^{(4)}_{\even}} =\frac{1}{\sqrt2} \left( \ket{\psi^{(2)}_{\even}}_{A'B_1'}\ket{\psi^{(2)}_{\even}}_{B_2'C'} - \ket{\psi^{(2)}_{\odd}}_{A'B_1'}\ket{\psi^{(2)}_{\odd}}_{B_2'C'} \right).
\label{eq:sm_even4_schmidt}
\end{eqnarray}
The two vectors $\ket{\psi^{(2)}_{\even}}$ and $\ket{\psi^{(2)}_{\odd}}$ are orthonormal. Eq.~(\ref{eq:sm_even4_schmidt}) is therefore a Schmidt decomposition with two nonzero Schmidt coefficients.
\end{proof}

This proves the advertised distinction. The first line, Eq.~(\ref{eq:sm_quotient_source}), is source-factorized in the quotient sense. The canonical representative, Eq.~(\ref{eq:sm_global_flag_network}), is not source-factorized in the usual tensor-product sense. The quotient equivalence relation is what reconciles the two descriptions.

Consequently, the reproduction of $T=6\sqrt2$ does not rescue ordinary real-amplitude tensor-product quantum mechanics. It replaces the source-independence assumption by quotient source independence. This is why Bell-network experiments can exclude the real-amplitude foil theory while leaving untouched any representation that is by construction isomorphic to complex quantum mechanics.

\section{Open quantum systems and allowed real channels}\label{sec:open}

The same hidden structure controls open-system dynamics. Let $V$ be the real Hilbert space with complex structure $\hat{J}$. For pure states, gauge equivalence is $v\sim e^{\alpha \hat{J}}v$. For density matrices, the corresponding equivalence is
\begin{eqnarray}
\hat{\rho}\sim e^{\alpha \hat{J}}\hat{\rho} e^{-\alpha \hat{J}}.
\label{eq:sm_rho_gauge}
\end{eqnarray}
A real channel is compatible with the quotient theory only if it descends to a well-defined channel on gauge equivalence classes.

\begin{theorem}[Descent criterion for real channels]\label{thm:channel_descent}
Let $\cE$ be a real completely positive trace-preserving map on density operators over $V$. It defines a well-defined channel on the quotient/gauge theory if and only if
\begin{eqnarray}
\hat{\rho}\sim\hat{\sigma}\quad\Rightarrow\quad \cE(\hat{\rho})\sim\cE(\hat{\sigma}).
\label{eq:sm_descent_criterion}
\end{eqnarray}
A sufficient condition is $\hat{J}$-covariance,
\begin{eqnarray}
\cE(e^{\alpha \hat{J}}\hat{\rho} e^{-\alpha \hat{J}}) =e^{\alpha \hat{J}}\cE(\hat{\rho})e^{-\alpha \hat{J}} \quad (\forall\alpha).
\label{eq:sm_channel_covariance}
\end{eqnarray}
\end{theorem}

\begin{proof}---A map on equivalence classes is well defined precisely when the output class does not depend on the chosen representative of the input class. This is Eq.~(\ref{eq:sm_descent_criterion}). If Eq.~(\ref{eq:sm_channel_covariance}) holds, then gauge-related inputs are sent to gauge-related outputs, so the descent criterion follows.
\end{proof}

In the Heisenberg picture, the observable algebra is
\begin{eqnarray}
\cA_{\hat{J}}=\{\hat{A}=\hat{A}^T:[\hat{A},\hat{J}]=0\}.
\label{eq:sm_AJ}
\end{eqnarray}
A channel is operationally compatible with the complex representation if its adjoint maps physical effects to physical effects,
\begin{eqnarray}
\cE^\dagger(\cA_{\hat{J}})\subseteq\cA_{\hat{J}}.
\label{eq:sm_heis_condition}
\end{eqnarray}
This condition is necessary if every physical measurement after the channel is to correspond to an allowed physical measurement before the channel.

\begin{proposition}[Realification of complex channels is allowed]\label{prop:complex_channel_allowed}
Let $\Phi$ be a complex CPTP map with Kraus operators $\hat{K}_\mu$,
\begin{eqnarray}
\Phi(\hat{\rho})=\sum_\mu \hat{K}_\mu\hat{\rho} \hat{K}_\mu^\dagger.
\end{eqnarray}
Define the realified map on trace-normalized real density matrices by
\begin{eqnarray}
\Phi_R(\hat{\sigma})=\sum_\mu T(\hat{K}_\mu)\hat{\sigma} T(\hat{K}_\mu)^T.
\label{eq:sm_realified_channel}
\end{eqnarray}
Then,
\begin{eqnarray}
\Phi_R(T_1(\hat{\rho}))=T_1(\Phi(\hat{\rho})),
\label{eq:sm_channel_commute}
\end{eqnarray}
and $\Phi_R$ preserves the $\hat{J}$-commutant structure.
\end{proposition}

\begin{proof}---The map $T$ is a real algebra homomorphism and maps adjoints to transposes. Hence
\begin{eqnarray}
T(\hat{K}_\mu)T_1(\hat{\rho})T(\hat{K}_\mu)^T=T_1(\hat{K}_\mu\hat{\rho} \hat{K}_\mu^\dagger),
\end{eqnarray}
with the same trace-normalizing factor contained in $T_1$. Summing over $\mu$ gives Eq.~(\ref{eq:sm_channel_commute}). Since every $T(\hat{K}_\mu)$ commutes with $\hat{J}$, the map preserves gauge-compatible states and physical effects.
\end{proof}

By contrast, a generic real channel on the enlarged flag space is not allowed. For example, a measurement interaction that couples to $\ketbra{0}{0}_F$ distinguishes $v$ from $\hat{J}v$ and violates Theorem~\ref{thm:sm_superselection}. A dephasing map in the flag basis similarly singles out the real-imaginary decomposition and is not a realification of a complex channel unless it acts trivially on all physical probabilities. Therefore, the open-system equivalence does not require an infinite auxiliary space in finite dimension; it requires a strong restriction on admissible environment couplings.

For infinite-dimensional systems, an additional analytic issue appears. Realification of bounded operators is straightforward, but unbounded operators require domain control: the domain of $\hat{A}$, the domain of its realification, and the action of $\hat{J}$ must be compatible. Thus, the finite-dimensional equivalence does not automatically prove a full infinite-dimensional theorem for all Hamiltonians, generators, and open-system semigroups. This is a limitation of the scope of the quotient construction, not a contradiction in its finite-dimensional version.

\section{Relation to previous and concurrent work}\label{sec:relation}

This section situates the Letter within the rapidly developing literature on real-number formulations of quantum mechanics. The point of the Letter is not to claim that real encodings are new, nor to deny the mathematical consistency of quotient or K\"ahler realifications. The point is narrower and, we believe, more operational: in the quotient-space construction of Ref.~\cite{SM-BarriosHita2026}, empirical equivalence is obtained only after one adds a distinguished complex structure $\hat{J}$, restricts the physical algebra and dynamics to the corresponding gauge-preserving sector, and composes systems with the balanced tensor product. Several previous and concurrent works emphasize adjacent aspects of this same situation.

\subsection{Previous studies}

\subsubsection{Older real encodings: Stueckelberg, Myrheim, rebits, and ubits} 

Stueckelberg's real-Hilbert-space formulation~\cite{SM-Stueckelberg1960} already contains the central structural lesson: a real Hilbert space by itself is not enough to reproduce complex quantum mechanics. One must also specify an operator playing the role of multiplication by $i$ and restrict physical operations accordingly. In modern language, this is the selection of the $\hat{J}$-commutant as the physical algebra. The $\hat{J}$-superselection theorem in Sec.~\ref{sec:superselection} is a coordinate-free way of restating why such a restriction is necessary. Without it, the real and imaginary components, or equivalently the flag, become observable.

Myrheim's construction~\cite{SM-Myrheim1999} treats complex Hilbert spaces as real Hilbert spaces and analyzes composite systems by restricting the ordinary real tensor product to a suitable physical subspace. In the bipartite case, this subspace is selected by a constraint involving $\hat{J}_A\otimes \hat{J}_B$, which is the subspace version of the balancing relation used in Eq.~\eqref{eq:sm_balanced_def}. The quotient construction of Ref.~\cite{SM-BarriosHita2026} is therefore close in spirit to Myrheim's approach: both remove the extra degrees of freedom of the ordinary real tensor product by imposing compatibility with the hidden complex structures. The difference is formal but important. Myrheim restricts to a subspace, whereas the quotient construction identifies vectors modulo the kernel of the inverse map. The latter preserves the statement that pure states are arbitrary unit vectors of the assigned quotient Hilbert space, but it does not remove the underlying $\hat{J}$-structure.

The rebit and ubit simulations of McKague, Mosca, and Gisin~\cite{SM-McKague2009}, Wootters~\cite{SM-Wootters2012}, and Aleksandrova, Borish, and Wootters~\cite{SM-Aleksandrova2013} also show that complex quantum systems can be simulated using real Hilbert spaces once an additional binary or universal real degree of freedom is supplied. These works are valuable because they make the encoding explicit. Our analysis differs in emphasis: a faithful encoding of complex quantum mechanics is not yet an independent real-amplitude theory. The operational question is whether the added real degree of freedom is accessible and whether the allowed physical algebra is the full real algebra or only the subalgebra compatible with the encoded complex structure.

\subsubsection{The Renou no-go result and the foil-theory perspective} 

Renou \textit{et al.}~\cite{SM-Renou2021} showed that standard complex quantum mechanics and ordinary real-amplitude quantum mechanics with the usual tensor product can make different predictions in network scenarios with independent sources. The subsequent experiments~\cite{SM-Li2022,SM-Chen2022,SM-Wu2022} tested this real-amplitude foil theory. These results do not establish that every possible real notation for complex quantum mechanics is experimentally excluded. Rather, they exclude a particular operational theory in which the scalar field is real and the composite rule is the ordinary real tensor product.

This is precisely the interpretation advocated in the foil-theory analysis of Ying \textit{et al.}~\cite{SM-YingEtAl2025}. In that view, the Renou result is best understood as an experimentally meaningful comparison between standard quantum theory and a carefully specified alternative, real-amplitude quantum theory. Our Letter is aligned with this methodological point. The quotient theory is not the same foil theory. It is empirically equivalent to complex quantum mechanics because it changes the composition rule and the observable algebra. Therefore, its survival does not make the Renou-type tests irrelevant; it clarifies which alternative was being tested.

\subsubsection{Barrios Hita \textit{et al.}: quotient real quantum mechanics} 

The immediate target of the Letter is the quotient-space construction of Barrios Hita \textit{et al.}~\cite{SM-BarriosHita2026}. That work correctly observes that the ordinary real tensor product fails to encode the complex phase-balancing freedom of product states. It then introduces a quotient of the real tensor product by the kernel of the product of inverse single-system maps. The result is one-to-one with complex quantum mechanics and hence reproduces all complex predictions.

Our contribution is to identify the operational content of this quotient. Section~\ref{sec:balanced} proves that it is the balanced tensor product
\begin{eqnarray}
V_A\bal V_B=(V_A\otimes_{\R}V_B)/\langle \hat{J}_A v\otimes w-v\otimes \hat{J}_B w\rangle,
\end{eqnarray}
which is canonically isomorphic to the realification of $H_A\otimes_{\C}H_B$. Thus, the composite rule is not the ordinary real tensor product made more physical by locality alone; it is canonically (or naturally) isomorphic to the complex tensor product written as a real quotient. Similarly, Sec.~\ref{sec:superselection} shows that the single-system flag is not an ordinary real subsystem. The generic real flag measurements must be forbidden by a $\hat{J}$-superselection rule.

This is why we do not describe Ref.~\cite{SM-BarriosHita2026} as inconsistent. We describe it as a real representation of complex quantum mechanics. The distinction matters because the conclusion ``complex numbers are a matter of convenience'' can mean two different things. Complex coordinates are indeed dispensable. The complex structure encoded by $\hat{J}$, by the $\hat{J}$-commutant, and by the balanced tensor product is not thereby eliminated.

\subsubsection{Hoffreumon--Woods: modified composition and operational independence} 

Hoffreumon and Woods~\cite{SM-HoffreumonWoods2025} defend a real-number formulation by deriving a nonstandard combination rule from a real single-system starting point and by emphasizing representation locality. Their construction is close in spirit to quotient or K\"ahler realifications: the usual Kronecker product is not treated as the correct real composition rule for the empirically equivalent theory. Our analysis is compatible with that mathematical move, but interprets its content differently. Replacing the ordinary real tensor product by a rule isomorphic to complex composition is precisely the point at which the complex structure re-enters the real formalism.

Their later work~\cite{SM-HoffreumonWoods2026} draws a sharp distinction between product-state independence and operational independence. They argue that if independence is imposed only through absence of observable cross-source correlations, real quantum theory cannot be experimentally falsified in finite networks and sequential protocols. This is closely related to the distinction in Sec.~\ref{sec:bell}: the quotient expression $T_1(\hat{\rho}_{AB_1})\bal T_1(\hat{\rho}_{B_2C})$ is independent in the quotient sense, while its canonical representative contains a global parity flag and is not source-factorizable in the ordinary ambient tensor product. Our focus, however, is not to decide which independence notion should be adopted as a foundational postulate. It is to make explicit the algebraic price of the empirical equivalence: nonstandard composition, hidden $\hat{J}$-structure, and restricted operations.

The comment by Kalarde, Xu, and Renou~\cite{SM-KalardeXuRenou2026} is relevant in this context because it challenges the generality of a proposed physical postulate in Ref.~\cite{SM-HoffreumonWoods2026} by testing it against fermionic information theory. We do not rely on their specific fermionic argument. Their work nevertheless supports a broader caution that postulates introduced to select one real formulation over another are nontrivial physical assumptions, not neutral mathematical bookkeeping.

\subsubsection{Feng--Ren--Vedral: nonlocal maps and hidden nonlocal degrees of freedom} 

Feng, Ren, and Vedral~\cite{SM-FengRenVedral2025} argue that real-number quantum theories compatible with independent sources require a nonlocal map, and that complex quantum theory is equivalent to a real theory with hidden nonlocal degrees of freedom. This is one of the closest perspectives to ours. In our terminology, the nonlocal-looking feature is visible when canonical representatives are chosen in the ambient tensor product of local flags. The four-party even flag in Eq.~\eqref{eq:sm_global_flag_network} has Schmidt rank two across the two-source split. It is source-factorized only after quotienting.

The difference is mainly one of formulation. Feng--Ren--Vedral emphasize nonlocality of the map required to combine independent systems. We give an algebraic diagnosis of the same phenomenon: the map is nonstandard because the correct composition is balanced over $\hat{J}$, not tensoring over $\R$. This also clarifies the two possible interpretations. If the flags are physical ancillary systems, the canonical representatives exhibit global flag constraints. If the flags are gauge variables, there is no physical nonlocal resource, but the theory is just a gauge-redundant real presentation of complex quantum mechanics.

\subsubsection{Weilenmann--Gisin--Sekatski: partial independence and resource cost} 

Weilenmann, Gisin, and Sekatski~\cite{SM-WeilenmannGisinSekatski2025} revisit the source-independence assumption and show that partial independence is already sufficient to rule out ordinary real quantum theory experimentally. They also quantify the tradeoff between source independence and achievable Bell value, and identify real entangled resources sufficient to recover complex correlations in multi-source scenarios. This work complements our Sec.~\ref{sec:bell}. Where they provide quantitative bounds and resource costs, we identify the algebraic structure responsible for the resource: a global parity or phase-reference structure associated with the hidden complex operator $\hat{J}$.

In particular, their observation that simulating complex multi-source scenarios over real Hilbert spaces requires additional shared real resources is the resource-theoretic counterpart of our Eq.~\eqref{eq:sm_global_flag_network}. The quotient formulation hides this resource inside the equivalence relation. The canonical representative reveals it as a global flag state. Thus, their partial-independence analysis and our quotient-source-independence analysis point to the same conceptual divide: ordinary product-state independence and quotient/gauge independence are not the same assumption.

\subsubsection{K\"ahler and symplectic realifications} 

Volovich~\cite{SM-Volovich2025} formulates quantum mechanics on real K\"ahler spaces, emphasizing the equivalence between complex Hilbert spaces and real spaces equipped with metric, symplectic form, and compatible complex structure. Maioli, Curado, and Gazeau~\cite{SM-MaioliCuradoGazeau2026} develop a related real K\"ahler framework and introduce a symplectic composition rule that replaces the ordinary real Kronecker product, reproducing the standard complex predictions including the relevant network-Bell value. These works are mathematically close to the balanced-tensor-product analysis in Sec.~\ref{sec:balanced}.

Our disagreement is not with the mathematics. A K\"ahler realification is exactly what one expects when a complex Hilbert space is written over $\R$: the imaginary unit becomes a complex structure $\hat{J}$, the imaginary part of the inner product becomes a symplectic form, and the tensor product must respect these structures. Our interpretive point is that this is not a bare real Hilbert-space theory. It is a real Hilbert space plus geometric data equivalent to the complex Hilbert-space structure. Therefore, such works can be read as supporting the central distinction of the Letter: complex coordinates may be removed, but complex structure remains.

\subsubsection{Symmetry emergence of complex structure} 

Moretti and Oppio~\cite{SM-MorettiOppio2017} show, in a relativistic setting, that a natural complex structure can emerge from Poincar\'e symmetry in real Hilbert space. Their analysis is different in scope from ours: it concerns elementary relativistic systems and symmetry representations, whereas our Letter analyzes finite-dimensional quotient constructions and network locality. Nevertheless, the moral is consistent. A complex structure is not a superficial notational artifact; it can be selected by physical requirements such as symmetry, locality, and the observable algebra. Our $\hat{J}$-superselection theorem is the finite-dimensional operational version of the same theme.

\subsection{What is new in the present Letter}

\begin{table}[t]
\centering
\begin{adjustbox}{max width=\textwidth}
\begin{tabular}{llll}
\toprule
\makecell[l]{Work or approach} & \makecell[l]{Main point} & \makecell[l]{Key structure} & \makecell[l]{Relation to the present Letter} \\
\midrule
\makecell[l]{Stueckelberg; rebit/ubit\\encodings~\cite{SM-Stueckelberg1960,SM-McKague2009,SM-Wootters2012,SM-Aleksandrova2013}} & \makecell[l]{Complex amplitudes can be\\simulated over real spaces.} & \makecell[l]{Additional rebit/ubit or complex\\structure; restricted algebra.} & \makecell[l]{Ancestor of our $\hat{J}$-superselection\\diagnosis.} \\
\addlinespace
Myrheim~\cite{SM-Myrheim1999} & \makecell[l]{Complex composition can be represented\\inside a real tensor product.} & \makecell[l]{Physical subspace selected by\\constraints involving local $\hat{J}$'s.} & \makecell[l]{Subspace analogue of the\\quotient/balanced construction.} \\
\addlinespace
\makecell[l]{Renou \textit{et al.} and\\experiments~\cite{SM-Renou2021,SM-Li2022,SM-Chen2022,SM-Wu2022}} & \makecell[l]{Ordinary real-amplitude tensor-product\\theory differs from complex QT in networks.} & \makecell[l]{Real Hilbert spaces with ordinary\\$\otimes_{\R}$ and product source states.} & \makecell[l]{Tests the foil theory, not\\gauge-redundant real encodings.} \\
\addlinespace
Ying \textit{et al.}~\cite{SM-YingEtAl2025} & \makecell[l]{Complex-vs-real debate framed\\through foil theories.} & \makecell[l]{Symmetrized subtheories and causal\\compatibility gaps.} & \makecell[l]{Supports our distinction between\\RQT and real encodings of complex QT.} \\
\addlinespace
Barrios Hita \textit{et al.}~\cite{SM-BarriosHita2026} & \makecell[l]{Quotient real QM reproduces\\all complex predictions.} & \makecell[l]{Flag, $SO(2)$ gauge, quotient\\tensor product.} & \makecell[l]{Immediate target; we identify hidden\\$\hat{J}$ and balanced composition.} \\
\addlinespace
Hoffreumon--Woods~\cite{SM-HoffreumonWoods2025,SM-HoffreumonWoods2026} & \makecell[l]{Real QT can be empirically equivalent\\with altered composition or independence.} & \makecell[l]{Modified combination rule;\\operational independence.} & \makecell[l]{Compatible mathematically; we diagnose\\the added algebraic structure.} \\
\addlinespace
Feng--Ren--Vedral~\cite{SM-FengRenVedral2025} & \makecell[l]{Independent-source real theories require\\a nonlocal map or hidden nonlocal degrees.} & \makecell[l]{Nonlocal map in real composition.} & \makecell[l]{Closest critique; balanced tensor product\\gives the algebraic form.} \\
\addlinespace
Weilenmann--Gisin--Sekatski~\cite{SM-WeilenmannGisinSekatski2025} & \makecell[l]{Partial independence suffices to rule out\\ordinary real QT; resource costs quantified.} & \makecell[l]{Source-independence/Bell-value tradeoff;\\shared real resources.} & \makecell[l]{Quantitative counterpart of our\\global parity-flag observation.} \\
\addlinespace
Volovich; Maioli--Curado--Gazeau~\cite{SM-Volovich2025,SM-MaioliCuradoGazeau2026} & \makecell[l]{Real K\"ahler or symplectic formulations\\reproduce standard QT.} & \makecell[l]{Metric, symplectic form, complex\\structure, symplectic composition.} & \makecell[l]{Mathematically close; we interpret this as\\relocating complex structure.} \\
\addlinespace
Moretti--Oppio~\cite{SM-MorettiOppio2017} & \makecell[l]{Complex structure can emerge from\\Poincar\'e symmetry.} & \makecell[l]{Symmetry-selected $\hat{J}$ commuting\\with observables.} & \makecell[l]{Supports the view that $\hat{J}$ is\\physical structure, not notation.} \\
\addlinespace
Kalarde--Xu--Renou~\cite{SM-KalardeXuRenou2026} & \makecell[l]{Challenges generality of a postulate selecting\\an empirically equivalent real formulation.} & \makecell[l]{Fermionic information theory\\as a stress test.} & \makecell[l]{Reinforces that new real-composition\\postulates are substantive assumptions.} \\
\addlinespace
Present Letter & \makecell[l]{Quotient real QM is consistent but is\\complex QM in real gauge notation.} & \makecell[l]{$\hat{J}$-superselection, $\hat{J}$-covariant\\dynamics, balanced tensor product.} & \makecell[l]{Algebraic and operational diagnosis of\\the quotient construction.} \\
\bottomrule
\end{tabular}
\end{adjustbox}
\caption{Summary of the relationship between the present Letter and representative previous or concurrent works. The common theme is that empirically equivalent real formulations require structure beyond a bare real Hilbert space with the ordinary real tensor product.}
\label{tab:related_work}
\end{table}

The preceding works show that many ingredients of the debate are already present in the literature: real encodings, universal rebits, nonstandard composition rules, nonlocal-map critiques, partial-independence tradeoffs, K\"ahler realifications, and foil-theory methodology. The specific contribution of the present Letter is to combine these ingredients into a precise diagnosis of the quotient-space construction of Ref.~\cite{SM-BarriosHita2026}:
\begin{enumerate}
\item the single-system flag requires a $\hat{J}$-superselection rule, otherwise global phase becomes measurable;
\item the quotient tensor product is the balanced tensor product over the hidden complex structure and is isomorphic to the complex tensor product;
\item the locality postulate P4 constrains embeddings but does not derive the quotient embeddings;
\item source independence in the Renou network is changed from ordinary product-state factorization to quotient factorization;
\item open dynamics must be restricted to channels that descend to the quotient, for instance $\hat{J}$-covariant channels.
\end{enumerate}
Thus, our claim is not that quotient-space real quantum mechanics is mathematically wrong. It is that its empirical equivalence to complex quantum mechanics is achieved by relocating, not eliminating, the physical complex structure. To make the novelty boundary explicit, see Table~\ref{tab:related_work}.

\section{Summary of the logical structure}

The logical relations established above can be summarized as follows.

\begin{itemize}
\item First, the single-system flag is not an ordinary subsystem. If it were, generic real effects would measure it. To prevent the global phase of complex quantum mechanics from becoming observable, one must impose the $\hat{J}$-superselection rule
\begin{eqnarray}
\cA_{\phys}=\{\hat{A}=\hat{A}^T:[\hat{A},\hat{J}]=0\}.
\end{eqnarray}

\item Second, the composite-system quotient is the balanced tensor product
\begin{eqnarray}
V_A\bal V_B=(V_A\otimes_{\R}V_B)/\langle \hat{J}_A v\otimes w-v\otimes \hat{J}_B w\rangle,
\end{eqnarray}
which is canonically isomorphic to $(H_A\otimes_{\C}H_B)_{\R}$. Thus, the ordinary real tensor product has not been repaired; it has been replaced by the real form of the complex tensor product.

\item Third, P4 is a locality condition on specified embeddings, not a derivation of those embeddings. Both $\otimes_{\R}$ and $\bal$ can satisfy a local-triviality condition once their local embeddings are chosen.

\item Fourth, in Bell-network scenarios, the quotient construction reproduces complex predictions because source independence is understood in the quotient sense. Canonical representatives in the ambient local-flag tensor product need not be source-factorizable in the usual sense.

\item Fifth, open dynamics must also preserve the hidden complex structure. Generic real channels, measurements, or environment couplings are not part of the empirically equivalent theory.
\end{itemize}

These points do not refute the consistency of quotient-space real quantum mechanics. They refine its interpretation. The theory is standard complex quantum mechanics written over $\R$ with a gauge redundancy. Complex numbers can be hidden in this notation, but the complex structure remains part of the physical theory.


\begin{thebibliography}{22}%
\makeatletter
\providecommand \@ifxundefined [1]{%
 \@ifx{#1\undefined}
}%
\providecommand \@ifnum [1]{%
 \ifnum #1\expandafter \@firstoftwo
 \else \expandafter \@secondoftwo
 \fi
}%
\providecommand \@ifx [1]{%
 \ifx #1\expandafter \@firstoftwo
 \else \expandafter \@secondoftwo
 \fi
}%
\providecommand \natexlab [1]{#1}%
\providecommand \enquote  [1]{``#1''}%
\providecommand \bibnamefont  [1]{#1}%
\providecommand \bibfnamefont [1]{#1}%
\providecommand \citenamefont [1]{#1}%
\providecommand \href@noop [0]{\@secondoftwo}%
\providecommand \href [0]{\begingroup \@sanitize@url \@href}%
\providecommand \@href[1]{\@@startlink{#1}\@@href}%
\providecommand \@@href[1]{\endgroup#1\@@endlink}%
\providecommand \@sanitize@url [0]{\catcode `\\12\catcode `\$12\catcode
  `\&12\catcode `\#12\catcode `\^12\catcode `\_12\catcode `\%12\relax}%
\providecommand \@@startlink[1]{}%
\providecommand \@@endlink[0]{}%
\providecommand \url  [0]{\begingroup\@sanitize@url \@url }%
\providecommand \@url [1]{\endgroup\@href {#1}{\urlprefix }}%
\providecommand \urlprefix  [0]{URL }%
\providecommand \Eprint [0]{\href }%
\providecommand \doibase [0]{https://doi.org/}%
\providecommand \selectlanguage [0]{\@gobble}%
\providecommand \bibinfo  [0]{\@secondoftwo}%
\providecommand \bibfield  [0]{\@secondoftwo}%
\providecommand \translation [1]{[#1]}%
\providecommand \BibitemOpen [0]{}%
\providecommand \bibitemStop [0]{}%
\providecommand \bibitemNoStop [0]{.\EOS\space}%
\providecommand \EOS [0]{\spacefactor3000\relax}%
\providecommand \BibitemShut  [1]{\csname bibitem#1\endcsname}%
\let\auto@bib@innerbib\@empty
\bibitem [{\citenamefont {Stueckelberg}(1960)}]{Stueckelberg1960}%
  \BibitemOpen
  \bibfield  {author} {\bibinfo {author} {\bibfnamefont {E.~C.~G.}\
  \bibnamefont {Stueckelberg}},\ }\href@noop {} {\bibfield  {journal} {\bibinfo
   {journal} {Helv. Phys. Acta}\ }\textbf {\bibinfo {volume} {33}},\ \bibinfo
  {pages} {727} (\bibinfo {year} {1960})}\BibitemShut {NoStop}%
\bibitem [{\citenamefont {Myrheim}(1999)}]{Myrheim1999}%
  \BibitemOpen
  \bibfield  {author} {\bibinfo {author} {\bibfnamefont {J.}~\bibnamefont
  {Myrheim}},\ }\href@noop {} {\bibinfo {title} {Quantum mechanics on a real
  hilbert space}} (\bibinfo {year} {1999}),\ \Eprint
  {https://arxiv.org/abs/quant-ph/9905037} {arXiv:quant-ph/9905037}
  \BibitemShut {NoStop}%
\bibitem [{\citenamefont {McKague}\ \emph {et~al.}(2009)\citenamefont
  {McKague}, \citenamefont {Mosca},\ and\ \citenamefont {Gisin}}]{McKague2009}%
  \BibitemOpen
  \bibfield  {author} {\bibinfo {author} {\bibfnamefont {M.}~\bibnamefont
  {McKague}}, \bibinfo {author} {\bibfnamefont {M.}~\bibnamefont {Mosca}},\
  and\ \bibinfo {author} {\bibfnamefont {N.}~\bibnamefont {Gisin}},\ }\href
  {https://doi.org/10.1103/PhysRevLett.102.020505} {\bibfield  {journal}
  {\bibinfo  {journal} {Phys. Rev. Lett.}\ }\textbf {\bibinfo {volume} {102}},\
  \bibinfo {pages} {020505} (\bibinfo {year} {2009})}\BibitemShut {NoStop}%
\bibitem [{\citenamefont {Wootters}(2012)}]{Wootters2012}%
  \BibitemOpen
  \bibfield  {author} {\bibinfo {author} {\bibfnamefont {W.~K.}\ \bibnamefont
  {Wootters}},\ }\href {https://doi.org/10.1007/s10701-011-9601-3} {\bibfield
  {journal} {\bibinfo  {journal} {Found. Phys.}\ }\textbf {\bibinfo {volume}
  {42}},\ \bibinfo {pages} {19} (\bibinfo {year} {2012})}\BibitemShut {NoStop}%
\bibitem [{\citenamefont {Aleksandrova}\ \emph {et~al.}(2013)\citenamefont
  {Aleksandrova}, \citenamefont {Borish},\ and\ \citenamefont
  {Wootters}}]{Aleksandrova2013}%
  \BibitemOpen
  \bibfield  {author} {\bibinfo {author} {\bibfnamefont {A.}~\bibnamefont
  {Aleksandrova}}, \bibinfo {author} {\bibfnamefont {V.}~\bibnamefont
  {Borish}},\ and\ \bibinfo {author} {\bibfnamefont {W.~K.}\ \bibnamefont
  {Wootters}},\ }\href {https://doi.org/10.1103/PhysRevA.87.052106} {\bibfield
  {journal} {\bibinfo  {journal} {Phys. Rev. A}\ }\textbf {\bibinfo {volume}
  {87}},\ \bibinfo {pages} {052106} (\bibinfo {year} {2013})}\BibitemShut
  {NoStop}%
\bibitem [{\citenamefont {Bang}\ \emph {et~al.}(2026)\citenamefont {Bang},
  \citenamefont {Cho},\ and\ \citenamefont {Yee}}]{BangChoYee2026}%
  \BibitemOpen
  \bibfield  {author} {\bibinfo {author} {\bibfnamefont {J.}~\bibnamefont
  {Bang}}, \bibinfo {author} {\bibfnamefont {K.}~\bibnamefont {Cho}},\ and\
  \bibinfo {author} {\bibfnamefont {K.~H.}\ \bibnamefont {Yee}},\ }\href
  {https://doi.org/10.48550/arXiv.2601.14638} {\bibinfo {title} {Quantum
  interference needs convention: {Overlap-Determinability} and unified
  {No-Superposition} principle}} (\bibinfo {year} {2026}),\ \Eprint
  {https://arxiv.org/abs/2601.14638} {arXiv:2601.14638 [quant-ph]} \BibitemShut
  {NoStop}%
\bibitem [{\citenamefont {Renou}\ \emph {et~al.}(2021)\citenamefont {Renou},
  \citenamefont {Trillo}, \citenamefont {Weilenmann}, \citenamefont {Le},
  \citenamefont {Tavakoli}, \citenamefont {Gisin}, \citenamefont {Ac{\'{i}}n},\
  and\ \citenamefont {Navascu{\'e}s}}]{Renou2021}%
  \BibitemOpen
  \bibfield  {author} {\bibinfo {author} {\bibfnamefont {M.-O.}\ \bibnamefont
  {Renou}}, \bibinfo {author} {\bibfnamefont {D.}~\bibnamefont {Trillo}},
  \bibinfo {author} {\bibfnamefont {M.}~\bibnamefont {Weilenmann}}, \bibinfo
  {author} {\bibfnamefont {T.~P.}\ \bibnamefont {Le}}, \bibinfo {author}
  {\bibfnamefont {A.}~\bibnamefont {Tavakoli}}, \bibinfo {author}
  {\bibfnamefont {N.}~\bibnamefont {Gisin}}, \bibinfo {author} {\bibfnamefont
  {A.}~\bibnamefont {Ac{\'{i}}n}},\ and\ \bibinfo {author} {\bibfnamefont
  {M.}~\bibnamefont {Navascu{\'e}s}},\ }\href
  {https://doi.org/10.1038/s41586-021-04160-4} {\bibfield  {journal} {\bibinfo
  {journal} {Nature}\ }\textbf {\bibinfo {volume} {600}},\ \bibinfo {pages}
  {625} (\bibinfo {year} {2021})}\BibitemShut {NoStop}%
\bibitem [{\citenamefont {Ying}\ \emph {et~al.}(2025)\citenamefont {Ying},
  \citenamefont {Alanon}, \citenamefont {Centeno}, \citenamefont {Surace},
  \citenamefont {Ansanelli}, \citenamefont {Liu}, \citenamefont {Schmid},\ and\
  \citenamefont {Spekkens}}]{YingEtAl2025}%
  \BibitemOpen
  \bibfield  {author} {\bibinfo {author} {\bibfnamefont {Y.}~\bibnamefont
  {Ying}}, \bibinfo {author} {\bibfnamefont {M.~C.}\ \bibnamefont {Alanon}},
  \bibinfo {author} {\bibfnamefont {D.}~\bibnamefont {Centeno}}, \bibinfo
  {author} {\bibfnamefont {J.}~\bibnamefont {Surace}}, \bibinfo {author}
  {\bibfnamefont {M.~M.}\ \bibnamefont {Ansanelli}}, \bibinfo {author}
  {\bibfnamefont {R.}~\bibnamefont {Liu}}, \bibinfo {author} {\bibfnamefont
  {D.}~\bibnamefont {Schmid}},\ and\ \bibinfo {author} {\bibfnamefont {R.~W.}\
  \bibnamefont {Spekkens}},\ }\href@noop {} {\bibinfo {title} {On whether
  quantum theory needs complex numbers: The foil theories perspective}}
  (\bibinfo {year} {2025}),\ \Eprint {https://arxiv.org/abs/2506.08091}
  {arXiv:2506.08091 [quant-ph]} \BibitemShut {NoStop}%
\bibitem [{\citenamefont {Li}\ \emph {et~al.}(2022)\citenamefont {Li},
  \citenamefont {Mao}, \citenamefont {Weilenmann}, \citenamefont {Tavakoli},
  \citenamefont {Chen}, \citenamefont {Feng}, \citenamefont {Yang},
  \citenamefont {Renou}, \citenamefont {Trillo}, \citenamefont {Le},
  \citenamefont {Gisin}, \citenamefont {Ac{\'{i}}n}, \citenamefont
  {Navascu{\'e}s}, \citenamefont {Wang},\ and\ \citenamefont {Fan}}]{Li2022}%
  \BibitemOpen
  \bibfield  {author} {\bibinfo {author} {\bibfnamefont {Z.-D.}\ \bibnamefont
  {Li}}, \bibinfo {author} {\bibfnamefont {Y.-L.}\ \bibnamefont {Mao}},
  \bibinfo {author} {\bibfnamefont {M.}~\bibnamefont {Weilenmann}}, \bibinfo
  {author} {\bibfnamefont {A.}~\bibnamefont {Tavakoli}}, \bibinfo {author}
  {\bibfnamefont {H.}~\bibnamefont {Chen}}, \bibinfo {author} {\bibfnamefont
  {L.}~\bibnamefont {Feng}}, \bibinfo {author} {\bibfnamefont {S.-J.}\
  \bibnamefont {Yang}}, \bibinfo {author} {\bibfnamefont {M.-O.}\ \bibnamefont
  {Renou}}, \bibinfo {author} {\bibfnamefont {D.}~\bibnamefont {Trillo}},
  \bibinfo {author} {\bibfnamefont {T.~P.}\ \bibnamefont {Le}}, \bibinfo
  {author} {\bibfnamefont {N.}~\bibnamefont {Gisin}}, \bibinfo {author}
  {\bibfnamefont {A.}~\bibnamefont {Ac{\'{i}}n}}, \bibinfo {author}
  {\bibfnamefont {M.}~\bibnamefont {Navascu{\'e}s}}, \bibinfo {author}
  {\bibfnamefont {Z.}~\bibnamefont {Wang}},\ and\ \bibinfo {author}
  {\bibfnamefont {J.}~\bibnamefont {Fan}},\ }\href
  {https://doi.org/10.1103/PhysRevLett.128.040402} {\bibfield  {journal}
  {\bibinfo  {journal} {Phys. Rev. Lett.}\ }\textbf {\bibinfo {volume} {128}},\
  \bibinfo {pages} {040402} (\bibinfo {year} {2022})}\BibitemShut {NoStop}%
\bibitem [{\citenamefont {Chen}\ \emph {et~al.}(2022)\citenamefont {Chen},
  \citenamefont {Wang}, \citenamefont {Liu}, \citenamefont {Wang},
  \citenamefont {Ying}, \citenamefont {Shang}, \citenamefont {Wu},
  \citenamefont {Gong}, \citenamefont {Deng}, \citenamefont {Liang},
  \citenamefont {Zhang}, \citenamefont {Peng}, \citenamefont {Zhu},
  \citenamefont {Cabello}, \citenamefont {Lu},\ and\ \citenamefont
  {Pan}}]{Chen2022}%
  \BibitemOpen
  \bibfield  {author} {\bibinfo {author} {\bibfnamefont {M.-C.}\ \bibnamefont
  {Chen}}, \bibinfo {author} {\bibfnamefont {C.}~\bibnamefont {Wang}}, \bibinfo
  {author} {\bibfnamefont {F.-M.}\ \bibnamefont {Liu}}, \bibinfo {author}
  {\bibfnamefont {J.-W.}\ \bibnamefont {Wang}}, \bibinfo {author}
  {\bibfnamefont {C.}~\bibnamefont {Ying}}, \bibinfo {author} {\bibfnamefont
  {Z.-X.}\ \bibnamefont {Shang}}, \bibinfo {author} {\bibfnamefont
  {Y.}~\bibnamefont {Wu}}, \bibinfo {author} {\bibfnamefont {M.}~\bibnamefont
  {Gong}}, \bibinfo {author} {\bibfnamefont {H.}~\bibnamefont {Deng}}, \bibinfo
  {author} {\bibfnamefont {F.-T.}\ \bibnamefont {Liang}}, \bibinfo {author}
  {\bibfnamefont {Q.}~\bibnamefont {Zhang}}, \bibinfo {author} {\bibfnamefont
  {C.-Z.}\ \bibnamefont {Peng}}, \bibinfo {author} {\bibfnamefont
  {X.}~\bibnamefont {Zhu}}, \bibinfo {author} {\bibfnamefont {A.}~\bibnamefont
  {Cabello}}, \bibinfo {author} {\bibfnamefont {C.-Y.}\ \bibnamefont {Lu}},\
  and\ \bibinfo {author} {\bibfnamefont {J.-W.}\ \bibnamefont {Pan}},\ }\href
  {https://doi.org/10.1103/PhysRevLett.128.040403} {\bibfield  {journal}
  {\bibinfo  {journal} {Phys. Rev. Lett.}\ }\textbf {\bibinfo {volume} {128}},\
  \bibinfo {pages} {040403} (\bibinfo {year} {2022})}\BibitemShut {NoStop}%
\bibitem [{\citenamefont {Wu}\ \emph {et~al.}(2022)\citenamefont {Wu},
  \citenamefont {Jiang}, \citenamefont {Gu}, \citenamefont {Huang},
  \citenamefont {Bai}, \citenamefont {Sun}, \citenamefont {Zhang},
  \citenamefont {Gong}, \citenamefont {Mao}, \citenamefont {Zhong},
  \citenamefont {Chen}, \citenamefont {Zhang}, \citenamefont {Zhang},
  \citenamefont {Lu},\ and\ \citenamefont {Pan}}]{Wu2022}%
  \BibitemOpen
  \bibfield  {author} {\bibinfo {author} {\bibfnamefont {D.}~\bibnamefont
  {Wu}}, \bibinfo {author} {\bibfnamefont {Y.-F.}\ \bibnamefont {Jiang}},
  \bibinfo {author} {\bibfnamefont {X.-M.}\ \bibnamefont {Gu}}, \bibinfo
  {author} {\bibfnamefont {L.}~\bibnamefont {Huang}}, \bibinfo {author}
  {\bibfnamefont {B.}~\bibnamefont {Bai}}, \bibinfo {author} {\bibfnamefont
  {Q.-C.}\ \bibnamefont {Sun}}, \bibinfo {author} {\bibfnamefont
  {X.}~\bibnamefont {Zhang}}, \bibinfo {author} {\bibfnamefont {S.-Q.}\
  \bibnamefont {Gong}}, \bibinfo {author} {\bibfnamefont {Y.}~\bibnamefont
  {Mao}}, \bibinfo {author} {\bibfnamefont {H.-S.}\ \bibnamefont {Zhong}},
  \bibinfo {author} {\bibfnamefont {M.-C.}\ \bibnamefont {Chen}}, \bibinfo
  {author} {\bibfnamefont {J.}~\bibnamefont {Zhang}}, \bibinfo {author}
  {\bibfnamefont {Q.}~\bibnamefont {Zhang}}, \bibinfo {author} {\bibfnamefont
  {C.-Y.}\ \bibnamefont {Lu}},\ and\ \bibinfo {author} {\bibfnamefont {J.-W.}\
  \bibnamefont {Pan}},\ }\href {https://doi.org/10.1103/PhysRevLett.129.140401}
  {\bibfield  {journal} {\bibinfo  {journal} {Phys. Rev. Lett.}\ }\textbf
  {\bibinfo {volume} {129}},\ \bibinfo {pages} {140401} (\bibinfo {year}
  {2022})}\BibitemShut {NoStop}%
\bibitem [{\citenamefont {Barrios~Hita}\ \emph {et~al.}(2026)\citenamefont
  {Barrios~Hita}, \citenamefont {Trushechkin}, \citenamefont {Kampermann},
  \citenamefont {Epping},\ and\ \citenamefont {Bru{\ss}}}]{BarriosHita2026}%
  \BibitemOpen
  \bibfield  {author} {\bibinfo {author} {\bibfnamefont {P.}~\bibnamefont
  {Barrios~Hita}}, \bibinfo {author} {\bibfnamefont {A.}~\bibnamefont
  {Trushechkin}}, \bibinfo {author} {\bibfnamefont {H.}~\bibnamefont
  {Kampermann}}, \bibinfo {author} {\bibfnamefont {M.}~\bibnamefont {Epping}},\
  and\ \bibinfo {author} {\bibfnamefont {D.}~\bibnamefont {Bru{\ss}}},\ }\href
  {https://doi.org/10.1103/4k13-sdjh} {\bibfield  {journal} {\bibinfo
  {journal} {Phys. Rev. Lett.}\ }\textbf {\bibinfo {volume} {136}},\ \bibinfo
  {pages} {240202} (\bibinfo {year} {2026})}\BibitemShut {NoStop}%
\bibitem [{Note1()}]{Note1}%
  \BibitemOpen
  \bibinfo {note} {Further details, including a coordinate-free proof and the
  equivalence with the flag-kernel quotient, are given in Sec.~IV of the
  Supplemental Material.}\BibitemShut {Stop}%
\bibitem [{Note2()}]{Note2}%
  \BibitemOpen
  \bibinfo {note} {The general $N$-party canonical even/odd flag states used
  here are defined in Sec.~V of the Supplemental Material.}\BibitemShut {Stop}%
\bibitem [{\citenamefont {Hoffreumon}\ and\ \citenamefont
  {Woods}(2025)}]{HoffreumonWoods2025}%
  \BibitemOpen
  \bibfield  {author} {\bibinfo {author} {\bibfnamefont {T.}~\bibnamefont
  {Hoffreumon}}\ and\ \bibinfo {author} {\bibfnamefont {M.~P.}\ \bibnamefont
  {Woods}},\ }\href@noop {} {\bibinfo {title} {Quantum theory does not need
  complex numbers}} (\bibinfo {year} {2025}),\ \Eprint
  {https://arxiv.org/abs/2504.02808} {arXiv:2504.02808 [quant-ph]} \BibitemShut
  {NoStop}%
\bibitem [{\citenamefont {Hoffreumon}\ and\ \citenamefont
  {Woods}(2026)}]{HoffreumonWoods2026}%
  \BibitemOpen
  \bibfield  {author} {\bibinfo {author} {\bibfnamefont {T.}~\bibnamefont
  {Hoffreumon}}\ and\ \bibinfo {author} {\bibfnamefont {M.~P.}\ \bibnamefont
  {Woods}},\ }\href@noop {} {\bibinfo {title} {Quantum theory based on real
  numbers cannot be experimentally falsified}} (\bibinfo {year} {2026}),\
  \Eprint {https://arxiv.org/abs/2603.19208} {arXiv:2603.19208 [quant-ph]}
  \BibitemShut {NoStop}%
\bibitem [{\citenamefont {Kalarde}\ \emph {et~al.}(2026)\citenamefont
  {Kalarde}, \citenamefont {Xu},\ and\ \citenamefont
  {Renou}}]{KalardeXuRenou2026}%
  \BibitemOpen
  \bibfield  {author} {\bibinfo {author} {\bibfnamefont {F.~M.}\ \bibnamefont
  {Kalarde}}, \bibinfo {author} {\bibfnamefont {X.}~\bibnamefont {Xu}},\ and\
  \bibinfo {author} {\bibfnamefont {M.-O.}\ \bibnamefont {Renou}},\ }\href@noop
  {} {\bibinfo {title} {Comment on ``quantum theory based on real numbers
  cannot be experimentally falsified'': On the compatibility of physical
  principles with information theory for fermions}} (\bibinfo {year} {2026}),\
  \Eprint {https://arxiv.org/abs/2604.07425} {arXiv:2604.07425 [quant-ph]}
  \BibitemShut {NoStop}%
\bibitem [{\citenamefont {Feng}\ \emph {et~al.}(2025)\citenamefont {Feng},
  \citenamefont {Ren},\ and\ \citenamefont {Vedral}}]{FengRenVedral2025}%
  \BibitemOpen
  \bibfield  {author} {\bibinfo {author} {\bibfnamefont {T.}~\bibnamefont
  {Feng}}, \bibinfo {author} {\bibfnamefont {C.}~\bibnamefont {Ren}},\ and\
  \bibinfo {author} {\bibfnamefont {V.}~\bibnamefont {Vedral}},\ }\href@noop {}
  {\bibinfo {title} {Locality implies complex numbers in quantum mechanics}}
  (\bibinfo {year} {2025}),\ \Eprint {https://arxiv.org/abs/2504.07808}
  {arXiv:2504.07808 [quant-ph]} \BibitemShut {NoStop}%
\bibitem [{\citenamefont {Weilenmann}\ \emph {et~al.}(2025)\citenamefont
  {Weilenmann}, \citenamefont {Gisin},\ and\ \citenamefont
  {Sekatski}}]{WeilenmannGisinSekatski2025}%
  \BibitemOpen
  \bibfield  {author} {\bibinfo {author} {\bibfnamefont {M.}~\bibnamefont
  {Weilenmann}}, \bibinfo {author} {\bibfnamefont {N.}~\bibnamefont {Gisin}},\
  and\ \bibinfo {author} {\bibfnamefont {P.}~\bibnamefont {Sekatski}},\
  }\href@noop {} {\bibfield  {journal} {\bibinfo  {journal} {Phys. Rev. Lett.}\
  }\textbf {\bibinfo {volume} {135}},\ \bibinfo {pages} {180201} (\bibinfo
  {year} {2025})},\ \Eprint {https://arxiv.org/abs/2502.20102}
  {arXiv:2502.20102 [quant-ph]} \BibitemShut {NoStop}%
\bibitem [{\citenamefont {Volovich}(2025)}]{Volovich2025}%
  \BibitemOpen
  \bibfield  {author} {\bibinfo {author} {\bibfnamefont {I.}~\bibnamefont
  {Volovich}},\ }\href@noop {} {\bibinfo {title} {Real quantum mechanics in a
  k{\"a}hler space}} (\bibinfo {year} {2025}),\ \Eprint
  {https://arxiv.org/abs/2504.16838} {arXiv:2504.16838 [quant-ph]} \BibitemShut
  {NoStop}%
\bibitem [{\citenamefont {Maioli}\ \emph {et~al.}(2026)\citenamefont {Maioli},
  \citenamefont {Curado},\ and\ \citenamefont
  {Gazeau}}]{MaioliCuradoGazeau2026}%
  \BibitemOpen
  \bibfield  {author} {\bibinfo {author} {\bibfnamefont {A.~C.}\ \bibnamefont
  {Maioli}}, \bibinfo {author} {\bibfnamefont {E.~M.~F.}\ \bibnamefont
  {Curado}},\ and\ \bibinfo {author} {\bibfnamefont {J.-P.}\ \bibnamefont
  {Gazeau}},\ }\href@noop {} {\bibinfo {title} {Quantum mechanics over real
  numbers fully reproduces standard quantum theory}} (\bibinfo {year} {2026}),\
  \Eprint {https://arxiv.org/abs/2604.19482} {arXiv:2604.19482 [quant-ph]}
  \BibitemShut {NoStop}%
\bibitem [{\citenamefont {Moretti}\ and\ \citenamefont
  {Oppio}(2017)}]{MorettiOppio2017}%
  \BibitemOpen
  \bibfield  {author} {\bibinfo {author} {\bibfnamefont {V.}~\bibnamefont
  {Moretti}}\ and\ \bibinfo {author} {\bibfnamefont {M.}~\bibnamefont
  {Oppio}},\ }\href@noop {} {\bibfield  {journal} {\bibinfo  {journal} {Rev.
  Math. Phys.}\ }\textbf {\bibinfo {volume} {29}},\ \bibinfo {pages} {1750021}
  (\bibinfo {year} {2017})},\ \Eprint {https://arxiv.org/abs/1611.09029}
  {arXiv:1611.09029 [math-ph]} \BibitemShut {NoStop}%
\end{thebibliography}

\begin{thebibliography}{19}%
\makeatletter
\providecommand \@ifxundefined [1]{%
 \@ifx{#1\undefined}
}%
\providecommand \@ifnum [1]{%
 \ifnum #1\expandafter \@firstoftwo
 \else \expandafter \@secondoftwo
 \fi
}%
\providecommand \@ifx [1]{%
 \ifx #1\expandafter \@firstoftwo
 \else \expandafter \@secondoftwo
 \fi
}%
\providecommand \natexlab [1]{#1}%
\providecommand \enquote  [1]{``#1''}%
\providecommand \bibnamefont  [1]{#1}%
\providecommand \bibfnamefont [1]{#1}%
\providecommand \citenamefont [1]{#1}%
\providecommand \href@noop [0]{\@secondoftwo}%
\providecommand \href [0]{\begingroup \@sanitize@url \@href}%
\providecommand \@href[1]{\@@startlink{#1}\@@href}%
\providecommand \@@href[1]{\endgroup#1\@@endlink}%
\providecommand \@sanitize@url [0]{\catcode `\\12\catcode `\$12\catcode
  `\&12\catcode `\#12\catcode `\^12\catcode `\_12\catcode `\%12\relax}%
\providecommand \@@startlink[1]{}%
\providecommand \@@endlink[0]{}%
\providecommand \url  [0]{\begingroup\@sanitize@url \@url }%
\providecommand \@url [1]{\endgroup\@href {#1}{\urlprefix }}%
\providecommand \urlprefix  [0]{URL }%
\providecommand \Eprint [0]{\href }%
\providecommand \doibase [0]{https://doi.org/}%
\providecommand \selectlanguage [0]{\@gobble}%
\providecommand \bibinfo  [0]{\@secondoftwo}%
\providecommand \bibfield  [0]{\@secondoftwo}%
\providecommand \translation [1]{[#1]}%
\providecommand \BibitemOpen [0]{}%
\providecommand \bibitemStop [0]{}%
\providecommand \bibitemNoStop [0]{.\EOS\space}%
\providecommand \EOS [0]{\spacefactor3000\relax}%
\providecommand \BibitemShut  [1]{\csname bibitem#1\endcsname}%
\let\auto@bib@innerbib\@empty
\bibitem [{\citenamefont {Barrios~Hita}\ \emph {et~al.}(2026)\citenamefont
  {Barrios~Hita}, \citenamefont {Trushechkin}, \citenamefont {Kampermann},
  \citenamefont {Epping},\ and\ \citenamefont {Bru{\ss}}}]{SM-BarriosHita2026}%
  \BibitemOpen
  \bibfield  {author} {\bibinfo {author} {\bibfnamefont {P.}~\bibnamefont
  {Barrios~Hita}}, \bibinfo {author} {\bibfnamefont {A.}~\bibnamefont
  {Trushechkin}}, \bibinfo {author} {\bibfnamefont {H.}~\bibnamefont
  {Kampermann}}, \bibinfo {author} {\bibfnamefont {M.}~\bibnamefont {Epping}},\
  and\ \bibinfo {author} {\bibfnamefont {D.}~\bibnamefont {Bru{\ss}}},\
  }\href@noop {} {\bibfield  {journal} {\bibinfo  {journal} {Phys. Rev. Lett.}\
  }\textbf {\bibinfo {volume} {136}},\ \bibinfo {pages} {240202} (\bibinfo
  {year} {2026})}\BibitemShut {NoStop}%
\bibitem [{\citenamefont {Renou}\ \emph {et~al.}(2021)\citenamefont {Renou},
  \citenamefont {Trillo}, \citenamefont {Weilenmann}, \citenamefont {Le},
  \citenamefont {Tavakoli}, \citenamefont {Gisin}, \citenamefont {Ac{\'i}n},\
  and\ \citenamefont {Navascu{\'e}s}}]{SM-Renou2021}%
  \BibitemOpen
  \bibfield  {author} {\bibinfo {author} {\bibfnamefont {M.-O.}\ \bibnamefont
  {Renou}}, \bibinfo {author} {\bibfnamefont {D.}~\bibnamefont {Trillo}},
  \bibinfo {author} {\bibfnamefont {M.}~\bibnamefont {Weilenmann}}, \bibinfo
  {author} {\bibfnamefont {T.~P.}\ \bibnamefont {Le}}, \bibinfo {author}
  {\bibfnamefont {A.}~\bibnamefont {Tavakoli}}, \bibinfo {author}
  {\bibfnamefont {N.}~\bibnamefont {Gisin}}, \bibinfo {author} {\bibfnamefont
  {A.}~\bibnamefont {Ac{\'i}n}},\ and\ \bibinfo {author} {\bibfnamefont
  {M.}~\bibnamefont {Navascu{\'e}s}},\ }\href@noop {} {\bibfield  {journal}
  {\bibinfo  {journal} {Nature}\ }\textbf {\bibinfo {volume} {600}},\ \bibinfo
  {pages} {625} (\bibinfo {year} {2021})}\BibitemShut {NoStop}%
\bibitem [{\citenamefont {Stueckelberg}(1960)}]{SM-Stueckelberg1960}%
  \BibitemOpen
  \bibfield  {author} {\bibinfo {author} {\bibfnamefont {E.~C.~G.}\
  \bibnamefont {Stueckelberg}},\ }\href@noop {} {\bibfield  {journal} {\bibinfo
   {journal} {Helv. Phys. Acta}\ }\textbf {\bibinfo {volume} {33}},\ \bibinfo
  {pages} {727} (\bibinfo {year} {1960})}\BibitemShut {NoStop}%
\bibitem [{\citenamefont {Myrheim}(1999)}]{SM-Myrheim1999}%
  \BibitemOpen
  \bibfield  {author} {\bibinfo {author} {\bibfnamefont {J.}~\bibnamefont
  {Myrheim}},\ }\href@noop {} {\bibinfo {title} {Quantum mechanics on a real
  hilbert space}} (\bibinfo {year} {1999}),\ \Eprint
  {https://arxiv.org/abs/quant-ph/9905037} {arXiv:quant-ph/9905037}
  \BibitemShut {NoStop}%
\bibitem [{\citenamefont {McKague}\ \emph {et~al.}(2009)\citenamefont
  {McKague}, \citenamefont {Mosca},\ and\ \citenamefont
  {Gisin}}]{SM-McKague2009}%
  \BibitemOpen
  \bibfield  {author} {\bibinfo {author} {\bibfnamefont {M.}~\bibnamefont
  {McKague}}, \bibinfo {author} {\bibfnamefont {M.}~\bibnamefont {Mosca}},\
  and\ \bibinfo {author} {\bibfnamefont {N.}~\bibnamefont {Gisin}},\
  }\href@noop {} {\bibfield  {journal} {\bibinfo  {journal} {Phys. Rev. Lett.}\
  }\textbf {\bibinfo {volume} {102}},\ \bibinfo {pages} {020505} (\bibinfo
  {year} {2009})}\BibitemShut {NoStop}%
\bibitem [{\citenamefont {Wootters}(2012)}]{SM-Wootters2012}%
  \BibitemOpen
  \bibfield  {author} {\bibinfo {author} {\bibfnamefont {W.~K.}\ \bibnamefont
  {Wootters}},\ }\href@noop {} {\bibfield  {journal} {\bibinfo  {journal}
  {Found. Phys.}\ }\textbf {\bibinfo {volume} {42}},\ \bibinfo {pages} {19}
  (\bibinfo {year} {2012})}\BibitemShut {NoStop}%
\bibitem [{\citenamefont {Aleksandrova}\ \emph {et~al.}(2013)\citenamefont
  {Aleksandrova}, \citenamefont {Borish},\ and\ \citenamefont
  {Wootters}}]{SM-Aleksandrova2013}%
  \BibitemOpen
  \bibfield  {author} {\bibinfo {author} {\bibfnamefont {A.}~\bibnamefont
  {Aleksandrova}}, \bibinfo {author} {\bibfnamefont {V.}~\bibnamefont
  {Borish}},\ and\ \bibinfo {author} {\bibfnamefont {W.~K.}\ \bibnamefont
  {Wootters}},\ }\href@noop {} {\bibfield  {journal} {\bibinfo  {journal}
  {Phys. Rev. A}\ }\textbf {\bibinfo {volume} {87}},\ \bibinfo {pages} {052106}
  (\bibinfo {year} {2013})}\BibitemShut {NoStop}%
\bibitem [{\citenamefont {Li}\ \emph {et~al.}(2022)\citenamefont {Li},
  \citenamefont {Mao}, \citenamefont {Weilenmann}, \citenamefont {Tavakoli},
  \citenamefont {Chen}, \citenamefont {Feng}, \citenamefont {Yang},
  \citenamefont {Renou}, \citenamefont {Trillo}, \citenamefont {Le},
  \citenamefont {Gisin}, \citenamefont {Ac{\'i}n}, \citenamefont
  {Navascu{\'e}s}, \citenamefont {Wang},\ and\ \citenamefont
  {Fan}}]{SM-Li2022}%
  \BibitemOpen
  \bibfield  {author} {\bibinfo {author} {\bibfnamefont {Z.-D.}\ \bibnamefont
  {Li}}, \bibinfo {author} {\bibfnamefont {Y.-L.}\ \bibnamefont {Mao}},
  \bibinfo {author} {\bibfnamefont {M.}~\bibnamefont {Weilenmann}}, \bibinfo
  {author} {\bibfnamefont {A.}~\bibnamefont {Tavakoli}}, \bibinfo {author}
  {\bibfnamefont {H.}~\bibnamefont {Chen}}, \bibinfo {author} {\bibfnamefont
  {L.}~\bibnamefont {Feng}}, \bibinfo {author} {\bibfnamefont {S.-J.}\
  \bibnamefont {Yang}}, \bibinfo {author} {\bibfnamefont {M.-O.}\ \bibnamefont
  {Renou}}, \bibinfo {author} {\bibfnamefont {D.}~\bibnamefont {Trillo}},
  \bibinfo {author} {\bibfnamefont {T.~P.}\ \bibnamefont {Le}}, \bibinfo
  {author} {\bibfnamefont {N.}~\bibnamefont {Gisin}}, \bibinfo {author}
  {\bibfnamefont {A.}~\bibnamefont {Ac{\'i}n}}, \bibinfo {author}
  {\bibfnamefont {M.}~\bibnamefont {Navascu{\'e}s}}, \bibinfo {author}
  {\bibfnamefont {Z.}~\bibnamefont {Wang}},\ and\ \bibinfo {author}
  {\bibfnamefont {J.}~\bibnamefont {Fan}},\ }\href@noop {} {\bibfield
  {journal} {\bibinfo  {journal} {Phys. Rev. Lett.}\ }\textbf {\bibinfo
  {volume} {128}},\ \bibinfo {pages} {040402} (\bibinfo {year}
  {2022})}\BibitemShut {NoStop}%
\bibitem [{\citenamefont {Chen}\ \emph {et~al.}(2022)\citenamefont {Chen},
  \citenamefont {Wang}, \citenamefont {Liu}, \citenamefont {Wang},
  \citenamefont {Ying}, \citenamefont {Shang}, \citenamefont {Wu},
  \citenamefont {Gong}, \citenamefont {Deng}, \citenamefont {Liang},
  \citenamefont {Zhang}, \citenamefont {Peng}, \citenamefont {Zhu},
  \citenamefont {Cabello}, \citenamefont {Lu},\ and\ \citenamefont
  {Pan}}]{SM-Chen2022}%
  \BibitemOpen
  \bibfield  {author} {\bibinfo {author} {\bibfnamefont {M.-C.}\ \bibnamefont
  {Chen}}, \bibinfo {author} {\bibfnamefont {C.}~\bibnamefont {Wang}}, \bibinfo
  {author} {\bibfnamefont {F.-M.}\ \bibnamefont {Liu}}, \bibinfo {author}
  {\bibfnamefont {J.-W.}\ \bibnamefont {Wang}}, \bibinfo {author}
  {\bibfnamefont {C.}~\bibnamefont {Ying}}, \bibinfo {author} {\bibfnamefont
  {Z.-X.}\ \bibnamefont {Shang}}, \bibinfo {author} {\bibfnamefont
  {Y.}~\bibnamefont {Wu}}, \bibinfo {author} {\bibfnamefont {M.}~\bibnamefont
  {Gong}}, \bibinfo {author} {\bibfnamefont {H.}~\bibnamefont {Deng}}, \bibinfo
  {author} {\bibfnamefont {F.-T.}\ \bibnamefont {Liang}}, \bibinfo {author}
  {\bibfnamefont {Q.}~\bibnamefont {Zhang}}, \bibinfo {author} {\bibfnamefont
  {C.-Z.}\ \bibnamefont {Peng}}, \bibinfo {author} {\bibfnamefont
  {X.}~\bibnamefont {Zhu}}, \bibinfo {author} {\bibfnamefont {A.}~\bibnamefont
  {Cabello}}, \bibinfo {author} {\bibfnamefont {C.-Y.}\ \bibnamefont {Lu}},\
  and\ \bibinfo {author} {\bibfnamefont {J.-W.}\ \bibnamefont {Pan}},\
  }\href@noop {} {\bibfield  {journal} {\bibinfo  {journal} {Phys. Rev. Lett.}\
  }\textbf {\bibinfo {volume} {128}},\ \bibinfo {pages} {040403} (\bibinfo
  {year} {2022})}\BibitemShut {NoStop}%
\bibitem [{\citenamefont {Wu}\ \emph {et~al.}(2022)\citenamefont {Wu},
  \citenamefont {Jiang}, \citenamefont {Gu}, \citenamefont {Huang},
  \citenamefont {Bai}, \citenamefont {Sun}, \citenamefont {Zhang},
  \citenamefont {Gong}, \citenamefont {Mao}, \citenamefont {Zhong},
  \citenamefont {Chen}, \citenamefont {Zhang}, \citenamefont {Zhang},
  \citenamefont {Lu},\ and\ \citenamefont {Pan}}]{SM-Wu2022}%
  \BibitemOpen
  \bibfield  {author} {\bibinfo {author} {\bibfnamefont {D.}~\bibnamefont
  {Wu}}, \bibinfo {author} {\bibfnamefont {Y.-F.}\ \bibnamefont {Jiang}},
  \bibinfo {author} {\bibfnamefont {X.-M.}\ \bibnamefont {Gu}}, \bibinfo
  {author} {\bibfnamefont {L.}~\bibnamefont {Huang}}, \bibinfo {author}
  {\bibfnamefont {B.}~\bibnamefont {Bai}}, \bibinfo {author} {\bibfnamefont
  {Q.-C.}\ \bibnamefont {Sun}}, \bibinfo {author} {\bibfnamefont
  {X.}~\bibnamefont {Zhang}}, \bibinfo {author} {\bibfnamefont {S.-Q.}\
  \bibnamefont {Gong}}, \bibinfo {author} {\bibfnamefont {Y.}~\bibnamefont
  {Mao}}, \bibinfo {author} {\bibfnamefont {H.-S.}\ \bibnamefont {Zhong}},
  \bibinfo {author} {\bibfnamefont {M.-C.}\ \bibnamefont {Chen}}, \bibinfo
  {author} {\bibfnamefont {J.}~\bibnamefont {Zhang}}, \bibinfo {author}
  {\bibfnamefont {Q.}~\bibnamefont {Zhang}}, \bibinfo {author} {\bibfnamefont
  {C.-Y.}\ \bibnamefont {Lu}},\ and\ \bibinfo {author} {\bibfnamefont {J.-W.}\
  \bibnamefont {Pan}},\ }\href@noop {} {\bibfield  {journal} {\bibinfo
  {journal} {Phys. Rev. Lett.}\ }\textbf {\bibinfo {volume} {129}},\ \bibinfo
  {pages} {140401} (\bibinfo {year} {2022})}\BibitemShut {NoStop}%
\bibitem [{\citenamefont {Ying}\ \emph {et~al.}(2025)\citenamefont {Ying},
  \citenamefont {Alanon}, \citenamefont {Centeno}, \citenamefont {Surace},
  \citenamefont {Ansanelli}, \citenamefont {Liu}, \citenamefont {Schmid},\ and\
  \citenamefont {Spekkens}}]{SM-YingEtAl2025}%
  \BibitemOpen
  \bibfield  {author} {\bibinfo {author} {\bibfnamefont {Y.}~\bibnamefont
  {Ying}}, \bibinfo {author} {\bibfnamefont {M.~C.}\ \bibnamefont {Alanon}},
  \bibinfo {author} {\bibfnamefont {D.}~\bibnamefont {Centeno}}, \bibinfo
  {author} {\bibfnamefont {J.}~\bibnamefont {Surace}}, \bibinfo {author}
  {\bibfnamefont {M.~M.}\ \bibnamefont {Ansanelli}}, \bibinfo {author}
  {\bibfnamefont {R.}~\bibnamefont {Liu}}, \bibinfo {author} {\bibfnamefont
  {D.}~\bibnamefont {Schmid}},\ and\ \bibinfo {author} {\bibfnamefont {R.~W.}\
  \bibnamefont {Spekkens}},\ }\href@noop {} {\bibinfo {title} {On whether
  quantum theory needs complex numbers: The foil theories perspective}}
  (\bibinfo {year} {2025}),\ \Eprint {https://arxiv.org/abs/2506.08091}
  {arXiv:2506.08091 [quant-ph]} \BibitemShut {NoStop}%
\bibitem [{\citenamefont {Hoffreumon}\ and\ \citenamefont
  {Woods}(2025)}]{SM-HoffreumonWoods2025}%
  \BibitemOpen
  \bibfield  {author} {\bibinfo {author} {\bibfnamefont {T.}~\bibnamefont
  {Hoffreumon}}\ and\ \bibinfo {author} {\bibfnamefont {M.~P.}\ \bibnamefont
  {Woods}},\ }\href@noop {} {\bibinfo {title} {Quantum theory does not need
  complex numbers}} (\bibinfo {year} {2025}),\ \Eprint
  {https://arxiv.org/abs/2504.02808} {arXiv:2504.02808 [quant-ph]} \BibitemShut
  {NoStop}%
\bibitem [{\citenamefont {Hoffreumon}\ and\ \citenamefont
  {Woods}(2026)}]{SM-HoffreumonWoods2026}%
  \BibitemOpen
  \bibfield  {author} {\bibinfo {author} {\bibfnamefont {T.}~\bibnamefont
  {Hoffreumon}}\ and\ \bibinfo {author} {\bibfnamefont {M.~P.}\ \bibnamefont
  {Woods}},\ }\href@noop {} {\bibinfo {title} {Quantum theory based on real
  numbers cannot be experimentally falsified}} (\bibinfo {year} {2026}),\
  \Eprint {https://arxiv.org/abs/2603.19208} {arXiv:2603.19208 [quant-ph]}
  \BibitemShut {NoStop}%
\bibitem [{\citenamefont {Kalarde}\ \emph {et~al.}(2026)\citenamefont
  {Kalarde}, \citenamefont {Xu},\ and\ \citenamefont
  {Renou}}]{SM-KalardeXuRenou2026}%
  \BibitemOpen
  \bibfield  {author} {\bibinfo {author} {\bibfnamefont {F.~M.}\ \bibnamefont
  {Kalarde}}, \bibinfo {author} {\bibfnamefont {X.}~\bibnamefont {Xu}},\ and\
  \bibinfo {author} {\bibfnamefont {M.-O.}\ \bibnamefont {Renou}},\ }\href@noop
  {} {\bibinfo {title} {Comment on ``quantum theory based on real numbers
  cannot be experimentally falsified'': On the compatibility of physical
  principles with information theory for fermions}} (\bibinfo {year} {2026}),\
  \Eprint {https://arxiv.org/abs/2604.07425} {arXiv:2604.07425 [quant-ph]}
  \BibitemShut {NoStop}%
\bibitem [{\citenamefont {Feng}\ \emph {et~al.}(2025)\citenamefont {Feng},
  \citenamefont {Ren},\ and\ \citenamefont {Vedral}}]{SM-FengRenVedral2025}%
  \BibitemOpen
  \bibfield  {author} {\bibinfo {author} {\bibfnamefont {T.}~\bibnamefont
  {Feng}}, \bibinfo {author} {\bibfnamefont {C.}~\bibnamefont {Ren}},\ and\
  \bibinfo {author} {\bibfnamefont {V.}~\bibnamefont {Vedral}},\ }\href@noop {}
  {\bibinfo {title} {Locality implies complex numbers in quantum mechanics}}
  (\bibinfo {year} {2025}),\ \Eprint {https://arxiv.org/abs/2504.07808}
  {arXiv:2504.07808 [quant-ph]} \BibitemShut {NoStop}%
\bibitem [{\citenamefont {Weilenmann}\ \emph {et~al.}(2025)\citenamefont
  {Weilenmann}, \citenamefont {Gisin},\ and\ \citenamefont
  {Sekatski}}]{SM-WeilenmannGisinSekatski2025}%
  \BibitemOpen
  \bibfield  {author} {\bibinfo {author} {\bibfnamefont {M.}~\bibnamefont
  {Weilenmann}}, \bibinfo {author} {\bibfnamefont {N.}~\bibnamefont {Gisin}},\
  and\ \bibinfo {author} {\bibfnamefont {P.}~\bibnamefont {Sekatski}},\
  }\href@noop {} {\bibfield  {journal} {\bibinfo  {journal} {Phys. Rev. Lett.}\
  }\textbf {\bibinfo {volume} {135}},\ \bibinfo {pages} {180201} (\bibinfo
  {year} {2025})},\ \Eprint {https://arxiv.org/abs/2502.20102}
  {arXiv:2502.20102 [quant-ph]} \BibitemShut {NoStop}%
\bibitem [{\citenamefont {Volovich}(2025)}]{SM-Volovich2025}%
  \BibitemOpen
  \bibfield  {author} {\bibinfo {author} {\bibfnamefont {I.}~\bibnamefont
  {Volovich}},\ }\href@noop {} {\bibinfo {title} {Real quantum mechanics in a
  {Kahler} space}} (\bibinfo {year} {2025}),\ \Eprint
  {https://arxiv.org/abs/2504.16838} {arXiv:2504.16838 [quant-ph]} \BibitemShut
  {NoStop}%
\bibitem [{\citenamefont {Maioli}\ \emph {et~al.}(2026)\citenamefont {Maioli},
  \citenamefont {Curado},\ and\ \citenamefont
  {Gazeau}}]{SM-MaioliCuradoGazeau2026}%
  \BibitemOpen
  \bibfield  {author} {\bibinfo {author} {\bibfnamefont {A.~C.}\ \bibnamefont
  {Maioli}}, \bibinfo {author} {\bibfnamefont {E.~M.~F.}\ \bibnamefont
  {Curado}},\ and\ \bibinfo {author} {\bibfnamefont {J.-P.}\ \bibnamefont
  {Gazeau}},\ }\href@noop {} {\bibinfo {title} {Quantum mechanics over real
  numbers fully reproduces standard quantum theory}} (\bibinfo {year} {2026}),\
  \Eprint {https://arxiv.org/abs/2604.19482} {arXiv:2604.19482 [quant-ph]}
  \BibitemShut {NoStop}%
\bibitem [{\citenamefont {Moretti}\ and\ \citenamefont
  {Oppio}(2017)}]{SM-MorettiOppio2017}%
  \BibitemOpen
  \bibfield  {author} {\bibinfo {author} {\bibfnamefont {V.}~\bibnamefont
  {Moretti}}\ and\ \bibinfo {author} {\bibfnamefont {M.}~\bibnamefont
  {Oppio}},\ }\href@noop {} {\bibfield  {journal} {\bibinfo  {journal} {Rev.
  Math. Phys.}\ }\textbf {\bibinfo {volume} {29}},\ \bibinfo {pages} {1750021}
  (\bibinfo {year} {2017})},\ \Eprint {https://arxiv.org/abs/1611.09029}
  {arXiv:1611.09029 [math-ph]} \BibitemShut {NoStop}%
\end{thebibliography}
\end{document}